\documentclass[11pt,a4paper]{article}
\usepackage{amssymb}
\usepackage{amsmath}
\usepackage{amsfonts}
\usepackage{bbm}
\usepackage{amsthm}
\usepackage{mathrsfs}
\usepackage{hyperref}
\usepackage{color}
\usepackage[margin=2.41cm]{geometry}
\usepackage[all,cmtip,2cell]{xy}
\UseAllTwocells
\usepackage[utf8]{inputenc}
\usepackage{graphicx}
\usepackage{varwidth}
\usepackage{comment}

\usepackage{upgreek}
\usepackage{rotating}

\usepackage{tikz}
\usetikzlibrary{shapes.geometric}

\definecolor{darkred}{rgb}{0.8,0.1,0.1}
\hypersetup{
     colorlinks=false,         
     linkcolor=darkred,
     citecolor=blue,
}

\theoremstyle{plain}
\newtheorem{theo}{Theorem}[section]

\newtheorem{propo}[theo]{Proposition}
\newtheorem{cor}[theo]{Corollary}

\theoremstyle{definition}
\newtheorem{defi}[theo]{Definition}

\newtheorem{exx}[theo]{Example}
\newenvironment{ex}
  {\pushQED{\qed}\exx}
  {\popQED\endexx}

\newtheorem{remm}[theo]{Remark}
\newenvironment{rem}
  {\pushQED{\qed}\remm}
  {\popQED\endremm}

\numberwithin{equation}{section}

\def\nn{\nonumber}

\def\bbR{\mathbb{R}}
\def\bbC{\mathbb{C}}

\def\bbZ{\mathbb{Z}}

\def\bbQ{\mathbb{Q}}
\def\bbS{\mathbb{S}}
\def\bbL{\mathbb{L}}
\def\bbE{\mathbb{E}}

\def\Hom{\mathrm{Hom}}
\def\hom{\mathrm{hom}}

\def\Imm{\mathrm{Im}}
\def\Ker{\mathrm{Ker}}
\def\PSh{\mathrm{PSh}}

\def\id{\mathrm{id}}

\def\dd{\mathrm{d}}

\def\1{I}
\def\oone{\mathbbm{1}}
\def\op{\mathrm{op}}
\def\ev{\mathrm{ev}}

\def\Loc{\mathbf{Loc}}

\def\Ran{\operatorname{Ran}}
\def\hoRan{\operatorname{hoRan}}
\def\Locc{\mathbf{Loc}_{\text{\large $\diamond$}}^{}}
\def\Man{\mathbf{Man}}

\def\Cart{\mathbf{Cart}}

\def\Set{\mathbf{Set}}
\def\Alg{\mathbf{Alg}}

\def\Ch{\mathbf{Ch}}

\def\dgAlg{\mathbf{dgAlg}}

\def\CC{\mathbf{C}}
\def\DD{\mathbf{D}}

\def\MM{\mathbf{M}}

\def\Cat{\mathbf{Cat}}
\def\OCat{\mathbf{OrthCat}}

\def\Grpd{\mathbf{Grpd}}
\def\sSet{\mathbf{sSet}}

\def\Op{\mathbf{Op}}
\def\Open{\mathbf{Open}}

\def\PBun{\mathbf{PBun}}
\def\Disc{\mathbf{Disk}}

\def\QFT{\mathbf{QFT}}

\def\AAA{\mathfrak{A}}

\def\BBB{\mathfrak{B}}

\def\CCC{\mathfrak{C}}
\def\DDD{\mathfrak{D}}

\def\O{\mathcal{O}}
\def\P{\mathcal{P}}
\def\E{\mathcal{E}}

\def\holim{\mathrm{holim}}
\def\colim{\mathrm{colim}}
\def\hocolim{\mathrm{hocolim}}
\def\Sing{\mathrm{Sing}}

\newcommand\und[1]{\underline{#1}}
\newcommand\ovr[1]{\overline{#1}}

\DeclareMathOperator*{\Motimes}{\text{\raisebox{0.25ex}{\scalebox{0.8}{$\bigotimes$}}}}

\def\sk{\vspace{2mm}}

\makeatletter
\let\@fnsymbol\@alph
\makeatother

\title{%
Homotopy theory of algebraic quantum field theories
}

\author{%
Marco Benini$^{1,a}$, 
Alexander Schenkel$^{2,b}$\ and\
Lukas Woike$^{1,c}$\vspace{4mm}\\
{\small ${}^1$ Fachbereich Mathematik, Universit\"at Hamburg,}\\
{\small Bundesstr.~55, 20146 Hamburg, Germany.}\vspace{3mm}\\
{\small ${}^2$ School of Mathematical Sciences, University of Nottingham,}\\
{\small University Park, Nottingham NG7 2RD, United Kingdom.}\vspace{5mm}\\
{\small \begin{tabular}{ll}
Email: & ${}^a$~\texttt{marco.benini@uni-hamburg.de}\\
& ${}^b$~\texttt{alexander.schenkel@nottingham.ac.uk}\\
& ${}^c$~\texttt{lukas.jannik.woike@uni-hamburg.de}\vspace{3mm}
\end{tabular}
}
}

\date{January 2019}

\begin{document}

\maketitle

\begin{abstract}
\noindent Motivated by gauge theory, we develop a general framework for chain complex valued algebraic quantum field theories. Building upon our recent operadic approach to this subject, we show that the category of such theories carries a canonical model structure and explain the important conceptual and also practical consequences of this result. As a concrete application we provide a derived version of Fredenhagen's universal algebra construction, which is relevant e.g.\ for the BRST/BV formalism. We further develop a homotopy theoretical generalization of algebraic quantum field theory with a particular focus on the homotopy-coherent Einstein causality axiom. We provide examples of such homotopy-coherent theories via (1) smooth normalized cochain algebras on $\infty$-stacks, and (2) fiber-wise groupoid cohomology of a category fibered in groupoids with coefficients in a strict quantum field theory.
\end{abstract}

\vspace{2mm}

\paragraph*{Report no.:} ZMP-HH/18-11, Hamburger Beitr\"age zur Mathematik Nr.\ 738

\paragraph*{Keywords:} algebraic quantum field theory, gauge theory, BRST/BV formalism, model categories, colored operads, homotopy algebras, $E_\infty$-algebras, $\infty$-stacks

\paragraph*{MSC 2010:} 81Txx, 18D50, 18G55, 55U35

\tableofcontents




\section{\label{sec:intro}Introduction and summary} 
Algebraic quantum field theory is a far developed mathematical
framework to investigate quantum field theories on Lorentzian
spacetimes from a model-independent perspective.
Its origins lie in a seminal paper of Haag and Kastler dating back 
to the early 1960s \cite{HaagKastler}. More than 50 years
of developments in this field have led to a broad spectrum
of beautiful and deep mathematical results for quantum field theories,
reaching from rigorous statements about their scattering theory, over 
discovering intriguing features in their representation theory, to constructing
non-perturbative models in low spacetime dimensions.
We refer to \cite{AQFTbook} for a recent general overview
and also to \cite{Kawahigashi} for a more specialized
review of chiral conformal quantum field theories.
\sk

One of the major open problems in algebraic quantum field theory
is that there is currently no consensus on the description
of gauge theories in this framework. In particular, the important questions 
of {\em what} is a quantum gauge theory and {\em how} it 
differs from a theory without gauge symmetries are not
yet sufficiently well understood. We believe that the key
to address and solve these problems is to develop a generalization 
of algebraic quantum field theory that takes
into account the crucial {\em higher structures}
which are present in gauge theories. This has led us to initiate the
{\em homotopical algebraic quantum field theory program}, cf.\
\cite{BeniniSchenkelSzabo,BeniniSchenkel,BSSStack,BeniniSchenkelWoike,BSWinvolutions}.
From a technical perspective, this program is about
combining the conceptual/physical ideas of algebraic 
quantum field theory with modern techniques from {\em homotopy theory},
in the sense of model category theory (see e.g.\ \cite{Dwyer,Hovey,Hirschhorn}) 
or higher category theory (see e.g.\ \cite{Lurie,HA}), which are capable to
describe the relevant higher structures of gauge theories.
We would like to emphasize that higher structures
already play a major role in other mathematical approaches to
quantum field theory, most notably in topological quantum field theory (see e.g.\ \cite{TQFT})
and in the factorization algebra approach of Costello and Gwilliam \cite{CostelloGwilliam}. Even though
these frameworks provide some inspiration for our developments, they are unfortunately
not directly applicable to our specific problem. This is because we are interested in quantum field theories
on {\em Lorentzian} spacetimes, in contrast to theories on topological or Riemannian manifolds, and for those the 
causal structure on spacetime is intrinsically linked to algebraic properties of quantum field theory 
via the Einstein causality axiom.
\sk

Let us provide a brief non-technical explanation of what are
the higher structures in gauge theory that we are talking about.
Recall that in a classical gauge theory we have two different
kinds of data, namely gauge fields and gauge transformations between gauge fields.
Hence, when thinking of the ``space of fields'' in a gauge theory,
one should think of a kind of {\em higher space} that consists of points, 
which describe the gauge fields, and also arrows between these points, which 
describe the gauge transformations between gauge fields. More technically,
this means that the space of gauge fields is not modeled by a set,
but rather by a {\em groupoid} (i.e.\ a category with only invertible arrows), 
which is the prime example of a higher structure. 
Let us also mention that the smooth structure on such a higher space of fields can be encoded
by using the concept of a {\em stack}, which is loosely speaking a smooth groupoid.
See e.g.\ \cite{Schreiber,Hollander,Dugger} for the technical background 
on ($\infty$-)stacks and \cite{BSSStack} for a physical example given by the 
stack of non-Abelian Yang-Mills fields. It is important to emphasize
that the stack of gauge fields is much richer
than the traditional gauge orbit space, which is the quotient space
obtained by identifying all gauge equivalent gauge fields. 
This is because the stack does not only indicate which of the 
gauge fields are gauge equivalent, but it also captures higher information by 
counting in how many ways they are gauge equivalent. 
There is the following striking analogy with the homotopy theory of topological spaces: 
This higher information can be thought of as ``loops'' in the stack of gauge fields and hence
should be interpreted as the $1$st homotopy group of the stack, while the 
gauge orbit space is only its $0$th homotopy group. (The precise concept 
is that of sheaves of homotopy groups for stacks, cf.\ \cite{Hollander,Dugger}.)
A key observation, which justifies the necessity of taking into account such
higher structures, is that they are crucial for obtaining a local-to-global (i.e.\ descent) property 
of the gauge theory, see e.g.\ \cite{BeniniSchenkelSzabo} for a concrete demonstration by 
computations and also \cite{Dougherty,Teh} for more philosophical arguments.
\sk

These higher structures play an important role also in quantum gauge theories.
To get some intuition on their appearance, let us adopt for the moment the 
point of view of deformation quantization. In this approach the observable algebras of a quantum field
theory are obtained by a deformation of the algebras of functions on the 
classical spaces of fields. Because in a gauge theory we have higher spaces (i.e.\ stacks) of gauge fields,
it is natural to expect that the corresponding observable algebras will be some kind of {\em higher algebras}.
Let us notice that certain shadows of such higher algebraic structures 
are already well-known in the physics literature, where they go under the name {\em BRST/BV formalism}, see e.g.\
\cite{Hollands,FredenhagenRejzner,FredenhagenRejzner2} for the corresponding developments
in algebraic quantum field theory. The quantum observable algebras in this framework are described
by certain {\em differential graded algebras}, which one can interpret as a quantization of a
higher function algebra (Chevalley-Eilenberg algebra) on the formal neighborhood 
of a point of the stack of gauge fields. In this construction, the higher structures encoded 
in the stack of gauge fields get translated to the higher homology groups of the differential 
graded algebras of observables. In particular this means that it is an important characteristic feature of
quantum gauge theories that their observables are described by higher algebras 
instead of ordinary gauge invariant observable algebras, which would provide
an incomplete picture because they describe only the $0$th homologies
and neglect all higher structures. Unfortunately, even though the BRST/BV formalism is able to
capture some of the higher structures of a quantum gauge theory, it is intrinsically perturbative
because it considers only the formal neighborhood of
a point of the stack of gauge fields.
We propose in Section \ref{subsec:stacks} of this paper a certain
non-perturbative generalization of the Chevalley-Eilenberg algebra, i.e.\ the classical
BRST formalism, which we will obtain by developing a concept of smooth normalized 
cochain algebras on stacks.
\sk

The overall aim of this paper is to develop a model-independent framework
for algebraic quantum field theories with values in chain complexes.
As we have explained in the paragraphs above, the generalization from
ordinary algebras in vector spaces to higher algebras 
in chain complexes is motivated by the necessity to encode and describe
the crucial higher structures of quantum gauge theories within the framework
of algebraic quantum field theory. Defining chain complex valued algebraic 
quantum field theories is relatively straightforward and these
have already appeared in previous works related to the BRST/BV formalism
\cite{Hollands,FredenhagenRejzner,FredenhagenRejzner2}.
Concretely, given a category $\CC$ of spacetimes, a chain complex valued algebraic quantum field theory
is a functor $\AAA : \CC\to \dgAlg(k)$ to the category of differential graded
algebras (over a commutative ring $k$) that satisfies some physically motivated properties 
such as the Einstein causality axiom and the time-slice axiom.
There is however the following subtle but essential point that has 
been mostly ignored in previous works: A common feature of all scenarios
involving higher structures is that the correct way to compare two objects
is not via isomorphisms, but via {\em weak equivalences}, which can be more general than
isomorphisms. For example, the weak equivalences between groupoids are equivalences of their underlying 
categories, the weak equivalences between topological spaces are weak homotopy equivalences, 
and the weak equivalences between chain complexes are quasi-isomorphisms,
i.e.\ maps between chain complexes that induce an isomorphism in homology.
We would like to stress the direct practical significance of weak equivalences: The usual technique of adding 
auxiliary fields and performing a gauge fixing in the BRST/BV formalism 
\cite{Hollands,FredenhagenRejzner,FredenhagenRejzner2}
is the prime example of a weak equivalence between chain complexes 
of observables that is in general not an isomorphism.
\sk

Working with categories that are endowed with a notion of weak equivalences
requires special care. The main reason is that ordinary categorical concepts and constructions,
such as functors, will generically fail to preserve weak equivalences, which
can lead to major issues because weakly equivalent objects in such situations are regarded
as ``being the same''. These problems can be solved systematically 
by using more sophisticated concepts describing category theory with a certain homotopy theoretical 
flavor, such as model category theory \cite{Dwyer,Hovey,Hirschhorn} 
or higher category theory \cite{Lurie,HA}. These powerful frameworks 
provide a consistent approach to describe categories with weak equivalences 
together with categorical constructions that are compatible with the weak equivalences.
\sk

The main aim of this paper is to show that there exists a canonical model structure on the
category of chain  complex valued algebraic quantum field theories and to explain 
the immense relevance of such structure both from a conceptual and from a more practical point of view.
The key ingredient for these developments is our recent result in \cite{BeniniSchenkelWoike} 
that algebraic quantum field  theories admit a description in terms of algebras over a 
suitable colored operad.
\sk

Now let us describe in more detail the content of the present paper.
In Section \ref{sec:prelim} we recall some basic definitions and results about 
the model category $\Ch(k)$ of chain complexes of $k$-modules, 
colored operads and their homotopy theory. This will be
particularly useful to fix our notations and also to make our article sufficiently
self-contained for readers who have only little experience with 
operads and/or homotopy theory. 
After these purely mathematical preliminaries, we focus in 
Section \ref{sec:strict} on $\Ch(k)$-valued algebraic quantum field theories.
Following \cite{BeniniSchenkelWoike}, we shall adopt a 
very broad and flexible definition of algebraic quantum field theory, which
encompasses various related concepts such as traditional algebraic
quantum field theory on a fixed spacetime \cite{HaagKastler}, 
locally covariant quantum field theory on all spacetimes \cite{Brunetti,FewsterVerch}
and chiral conformal quantum field theory on the circle \cite{Kawahigashi}.
In our framework each of these specific scenarios corresponds to a choice of 
{\em orthogonal category} $\ovr{\CC} = (\CC,\perp)$, cf.\ Definition \ref{def:OCat}.
To every orthogonal category $\ovr{\CC}$ we assign the category $\QFT(\ovr{\CC})$ 
of $\Ch(k)$-valued quantum field theories on $\ovr{\CC}$,
which is a full subcategory of the functor category $\dgAlg(k)^\CC$
consisting of all functors $\AAA : \CC\to \dgAlg(k)$ that satisfy the
{\em $\perp$-commutativity axiom}, cf.\ Definition \ref{def:QFTcats}.
Making use of the crucial result of \cite{BeniniSchenkelWoike} that $\QFT(\ovr{\CC})$ 
is the category of algebras over a suitable colored 
operad $\O_{\ovr{\CC}}$, we can endow the category $\QFT(\ovr{\CC})$
with a canonical model structure.
The corresponding weak equivalences admit a simple description in terms
of natural transformations $\zeta : \AAA\Rightarrow \BBB$ between 
theories $\AAA,\BBB : \CC\to\dgAlg(k)$ whose components 
$\zeta_c : \AAA(c)\to \BBB(c)$ are quasi-isomorphisms of chain complexes.
In particular, our rigorous concept of weak equivalences coincides
with the less formal notions used in practical constructions within the 
BRST/BV formalism, cf.\ \cite{Hollands,FredenhagenRejzner,FredenhagenRejzner2}.
As a non-trivial application of this model structure, we discuss 
Fredenhagen's universal algebra construction \cite{Fre1,Fre2,Fre3}
for $\Ch(k)$-valued quantum field theories from the perspective of
{\em derived functors} (cf.\ \cite{Dwyer,Hovey,Riehl}). We shall
present simple toy-models which show that
deriving the universal algebra construction 
is necessary to obtain a construction that 
preserves weak equivalences and that produces
the correct global gauge theory observables on arbitrary spacetimes.
The technical details are presented in Appendix \ref{app:extension}.
\sk

In Section \ref{sec:resolution} we explore a natural homotopy theoretical
generalization of algebraic quantum field theory.
These considerations are inspired by the concept of {\em homotopy algebras over operads},
which describe homotopy-coherent algebraic structures, such as
$A_\infty$-algebras and $E_\infty$-algebras.
In general, the homotopy algebras over an operad $\O$ are defined as 
algebras over a $\Sigma$-cofibrant resolution $\O_\infty\to \O$ of the operad.
In Section \ref{subsec:Sigmaresolution} we study 
$\Sigma$-cofibrant resolutions ${\O_{\ovr{\CC}}}_\infty \to \O_{\ovr{\CC}}$
of our algebraic quantum field theory operads and thereby
introduce {\em homotopy quantum field theories} 
given by algebras over the resolved colored operad ${\O_{\ovr{\CC}}}_\infty$.
As usual, there exist many different $\Sigma$-cofibrant resolutions 
${\O_{\ovr{\CC}}}_\infty \to \O_{\ovr{\CC}}$, i.e.\ different variants 
of homotopy quantum field theories. We prove that each of these
variants defines a Quillen equivalent model category of
homotopy quantum field theories and also that 
the identity $\id : \O_{\ovr{\CC}}\to \O_{\ovr{\CC}}$
is an instance of a $\Sigma$-cofibrant resolution.
The latter should be interpreted as a strictification 
result for homotopy quantum field theories to strict
algebraic quantum field theories in the sense of Section \ref{sec:strict}.
Let us stress that this strictification result {\em does not} mean that
homotopy quantum field theories are useless in practice,
because strictifications are extremely hard to compute and
hence it can be convenient to work with non-strict homotopy quantum field theories.
Inspired by our examples of non-strict homotopy quantum field theories 
in Sections \ref{subsec:stacks} and \ref{sec:orbifoldization}, we describe
in Section \ref{subsec:Einftyresolution} a specific $\Sigma$-cofibrant
resolution $\O_{\ovr{\CC}} \otimes \E_\infty \to \O_{\ovr{\CC}}$ that is obtained
by a component-wise tensor product of our quantum field theory operad 
with the Barratt-Eccles $E_\infty$-operad from \cite{BergerFresse}.
The resulting homotopy quantum field theories satisfy a homotopy-coherent
$\perp$-commutativity axiom. In particular, every $\E_\infty$-algebra valued functor 
$\AAA : \CC\to\Alg(\E_\infty)$ defines a homotopy quantum field theory 
over this resolution. In Section \ref{subsec:stacks} we construct examples of 
such functors by taking smooth normalized cochain algebras on $\infty$-stacks, 
which we interpret as a non-perturbative analog of the classical BRST formalism 
(Chevalley-Eilenberg algebra).
\sk

In Section \ref{sec:orbifoldization} we provide another class of examples of homotopy quantum
field theories over our resolution $\O_{\ovr{\CC}} \otimes \E_\infty  \to \O_{\ovr{\CC}}$.
Our construction is inspired by \cite{BeniniSchenkel} and it is based on the following idea:
We start with a category fibered in groupoids $\pi : \DD\to\CC$ over the category
of spacetimes $\CC$. The total category $\DD$ should be interpreted as a category of
spacetimes with background gauge fields, such as spin structures, bundles and connections. 
Then we take a strict $\Ch(k)$-valued quantum field theory 
on the total category $\DD$
and consider its underlying functor $\AAA : \DD\to\Ch(k)$. Over each spacetime
$c\in \CC$, there exists a groupoid $\pi^{-1}(c)$ of background fields and we
would like to take homotopy invariants of the corresponding groupoid actions on
the quantum field theory $\AAA$. We formalize this construction
in terms of a homotopy right Kan extension $\hoRan_\pi$ along the projection functor 
$\pi :\DD\to \CC$. Our main result is Theorem \ref{theo:hoRanQFT},
which proves that, after a very explicit strictification of the category fibered in groupoids
according to \cite{Hollander}, the functor $\hoRan_\pi \AAA : \CC\to\Ch(k)$ carries canonically
the structure of a homotopy quantum field theory over our resolution.
This construction does not only provide toy-models
of homotopy quantum field theories, but also seems to be relevant for
perturbative quantum gauge theories, cf.\  Example \ref{ex:fiberedcat}.
From a mathematical perspective, the homotopy quantum field theory
$\hoRan_\pi \AAA$ admits an interpretation in terms of fiber-wise groupoid cohomology
on $\pi :\DD\to\CC$ with coefficients in the strict quantum field theory $\AAA$, 
cf.\ Remark \ref{rem:groupoidcoho}.


\section{\label{sec:prelim}Preliminaries and notations}

\subsection{\label{subsec:complexes}Chain complexes}
Let $k$ be a commutative and unital ring. Throughout this paper we 
assume that the ring of rational numbers $\bbQ$ is a subring of $k$, 
i.e.\ $k\supseteq \bbQ$. This will considerably simplify our 
homotopy theoretical considerations, see \cite{HinichOriginal,Hinich}
and the subsections below for details. Notice that this assumption is satisfied in 
applications to quantum field theory, where $k$ will be either the field of complex numbers $\bbC$ 
or formal power series (in Planck's constant $\hbar$ and the coupling constants) in $\bbC$
if one is interested in perturbative theories.
\sk

We denote by $\Ch(k)$ the category of (possibly unbounded) {\em chain 
complexes of $k$-modules}. Recall that a chain complex
is a family of $k$-modules $\{V_n\,:\, n\in \bbZ\}$ together with a
differential, i.e.\ a family of $k$-linear maps $\{ \dd : V_n\to V_{n-1}\,:\, n\in\bbZ\}$
such that $\dd^2=0$. We often denote chain complexes simply by symbols like 
$V,W,Z \in\Ch(k)$ and use the same letter $\dd$ to denote various differentials.
A chain complex morphism $f : V\to W$ is a family of $k$-linear maps
$\{f : V_n\to W_n \,:\, n\in\bbZ\}$ that is compatible with the differentials, 
i.e.\ $\dd\, f = f\, \dd$.
\sk

The category $\Ch(k)$ carries the following
standard closed symmetric monoidal structure:
The tensor product $V\otimes W$ of two chain complexes
is defined as usual by
\begin{flalign}\label{eqn:tensorchaincomplex}
(V\otimes W)_n \,:=\, \bigoplus_{m \in\bbZ} V_m \otimes W_{n-m}\quad,
\end{flalign}
for all $n\in\bbZ$, together with the differential 
$\dd(v\otimes w) := \dd v \otimes w + (-1)^{m} \, v\otimes \dd w$,
for all $v\in V_m$ and $w\in W_{n-m}$. The monoidal unit
is $k\in\Ch(k)$, which we regard as a chain complex concentrated in degree $0$ 
with trivial differential $\dd=0$. The symmetric braiding 
$ \tau : V\otimes W\to W\otimes V$ is defined by
the usual sign rule $\tau(v\otimes w) := (-1)^{ml} \, w\otimes v$,
for all  $v\in V_m$ and $w\in W_l$. 
The internal hom object $[V,W]$ between 
two chain complexes is defined by
\begin{flalign}\label{eqn:internalhomchaincomplex}
[V,W]_n \,:=\, \prod_{m\in\bbZ}\hom_k^{}(V_m , W_{n+m})\quad,
\end{flalign}
for all $n\in \bbZ$, where $\hom_k^{}$ denotes the internal hom for $k$-modules (that
is the $k$-module of $k$-linear maps between $k$-modules), 
together with the differential $\dd(L) := \dd\,  L - (-1)^{n}\, L\, \dd$,
for all $(L: V_m\to W_{n+m}) \in \hom_k^{}(V_m,W_{n+m})$.
\sk

Let us further recall that a $\Ch(k)$-morphism $f : V\to W$
is called a {\em quasi-isomorphism} if it induces an isomorphism
in homology, i.e.\ $H_n(f) : H_n(V) \to H_n(W)$ is
an isomorphism of $k$-modules for each $n\in\bbZ$. 
The homology $k$-modules are defined as usual by the quotients
$H_n(V) := \Ker (\dd : V_n\to V_{n-1})/\Imm(\dd : V_{n+1}\to V_n)$,
for all $n\in\bbZ$. This notion of quasi-isomorphism in $\Ch(k)$
is part of the rich and powerful structure of a (symmetric monoidal)
{\em model category}. For the definition of a model category 
we refer to e.g.\ \cite{Dwyer,Hovey,Hirschhorn}. 
For an informal approach to our paper, however, it will be sufficient to keep in mind the
following core principle: A model category is a category together with
three distinguished classes of morphisms -- called {\em weak equivalences}, {\em fibrations}
and {\em cofibrations} -- that satisfy a list of axioms. These axioms are designed
in such a way that weak equivalences behave (with some technical care)
as good as isomorphisms, e.g.\ they are preserved by (derived) functors.
Hence, by using techniques from model category theory, we can consistently consider 
two objects as being the same not only when they are isomorphic, but also when they are
weakly equivalent. Such ideas and techniques are crucial for our development 
of a consistent framework for chain complex valued algebraic quantum field 
theories in this paper. The following result is proven in \cite[Sections 2.3 and 4.2]{Hovey}.
\begin{theo}\label{theo:Chmodelcat}
Define a morphism $f:V\to W$ in $\Ch(k)$ to be
\begin{itemize}
\item[(i)] a weak equivalence if it is a quasi-isomorphism, 
i.e.\ it induces an isomorphism in homology;
\item[(ii)] a fibration if it is surjective in each degree $f:V_n\to W_n$;
\item[(iii)] a cofibration if it has the left lifting property (cf.\ Remark \ref{rem:lifting})
with respect to all acyclic fibrations (i.e.\ all morphisms in $\Ch(k)$
that are both a fibration and a weak equivalence).
\end{itemize}
With these choices $\Ch(k)$ becomes a symmetric monoidal model category.
Moreover, the monoidal unit $k\in\Ch(k)$ is  a cofibrant object,
i.e.\ the unique morphism $0 \to k$ from the zero chain complex
is a cofibration.
\end{theo}

\begin{rem}\label{rem:lifting}
Let us recall that a morphism $f : V\to W$ is said to have the {\em left lifting property}
with respect to a morphism $f^\prime : V^\prime\to W^\prime$ 
if all commutative squares of the form
\begin{flalign}
\xymatrix{
\ar[d]_-{f}\ar[rr] V && \ar[d]^-{f^\prime}V^\prime\\
\ar@{-->}[rru]_-{f^{\prime\prime}}W \ar[rr] && W^\prime
}
\end{flalign}
admit a lift $f^{\prime\prime}$, i.e.\ the two triangles commute.
\end{rem}

\subsection{\label{subsec:operads}Homotopy theory of colored operads}
Let us start with a very brief review of some relevant aspects
of the theory of colored operads in a general bicomplete closed 
symmetric monoidal category $\MM$. (We will later take 
$\MM = \Ch(k)$, but the definitions below are easily stated
for general $\MM$.) We refer to \cite{BergerMoerdijk,Yau,YauQFT} 
and also to our previous paper \cite{BeniniSchenkelWoike} 
for a more extensive introduction. Let $\CCC\in\Set$
be a non-empty set. We refer to elements in $\CCC$ as colors. 
A {\em $\CCC$-colored operad} $\O$ with values in $\MM$ is given by the following data:
\begin{itemize}
\item[(i)] for each $n\geq 0$ and $n+1$-tuple of colors 
$(\und{c},t) := ((c_1,\dots,c_n),t)\in\CCC^{n+1}$, an object 
\begin{flalign}
\O\big(\substack{t\\\und{c}}\big)\in\MM\quad;
\end{flalign}

\item[(ii)] for each color $c\in\CCC$, an $\MM$-morphism
(called {\em operadic unit}) from the monoidal unit
\begin{flalign}
\oone : I \longrightarrow \O\big(\substack{c\\c}\big)\quad;
\end{flalign}

\item[(iii)] for each $n>0$, $n+1$-tuple of colors 
$(\und{a},t)\in \CCC^{n+1}$ and $k_i+1$-tuples of colors 
$(\und{b}_i,a_i)\in \CCC^{k_i +1}$, for $i=1,\dots,n$, an $\MM$-morphism
(called {\em operadic composition})
\begin{flalign}
\gamma  :  \O\big(\substack{t\\\und{a}}\big) \otimes\bigotimes_{i=1}^n 
\O\big(\substack{a_i\\\und{b}_i}\big) \longrightarrow
\O\big(\substack{t\\ \und{b}}\big)\quad,
\end{flalign}
where $\und{b} := (\und{b}_1, \dots, \und{b}_n)$ denotes
concatenation of tuples;

\item[(iv)] for each $n\geq 0$, $n+1$-tuple of colors
$(\und{c},t) \in\CCC^{n+1}$ and permutation $\sigma\in\Sigma_n$,
an $\MM$-morphism
\begin{flalign}\label{eqn:operadequivariances}
\O(\sigma) : \O\big(\substack{t\\ \und{c}}\big) \longrightarrow \O\big(\substack{t\\ \und{c}\sigma}\big)\quad,
\end{flalign}
where $\und{c}\sigma = (c_{\sigma(1)},\dots,c_{\sigma(n)})$ is given by right permutation.
\end{itemize}
These four data are required to satisfy compatibility conditions:
1.)~The morphisms in (iv) define right $\Sigma_n$-actions, i.e.\ $\O(\sigma^\prime)  \,\O(\sigma)
= \O(\sigma \sigma^\prime)$. 2.)~The operadic compositions in (iii) are equivariant
under these actions of the permutation groups. 3.)~The operadic compositions are 
associative and also unital with respect to the operadic units in (ii).
(See e.g.\ \cite[Definition 11.2.1]{Yau} for
the details.) The axioms for $\CCC$-colored operads 
admit the following intuitive interpretation:
$\O\big(\substack{t\\ \und{c}}\big)$ in (i) is interpreted
as the object of $n$-ary operations with input profile $\und{c}= (c_1,\dots,c_n)$ 
and target $t$. The morphisms in (iv) allow us to permute the input colors
and the operadic compositions in (iii) allow us to compose such operations. The operadic
units in (ii) should be regarded as the identity operations for each color $c\in\CCC$.
\sk

For every non-empty set of colors $\CCC\in\Set$, 
the $\CCC$-colored operads with values in $\MM$
form a category, which we denote by $\Op_\CCC(\MM)$. A morphism
$\phi : \O\to \P$ between two $\CCC$-colored operads is a collection
of $\MM$-morphisms 
\begin{flalign}
\phi :  \O\big(\substack{t\\ \und{c}}\big) \longrightarrow \P\big(\substack{t \\ \und{c}}\big)\quad,
\end{flalign}
for all $n\geq 0$ and $(\und{c},t)\in\CCC^{n+1}$, that is compatible with the permutation
actions $\phi\, \O(\sigma) = \P(\sigma)\,\phi$, the operadic compositions 
$\phi \,\gamma^\O=\gamma^\P\,\big(\phi \otimes\Motimes_{i=1}^n \phi\big) $ 
and the operadic units $\phi\,\oone^\O = \oone^\P$.
\sk

Now let us focus on the case of chain complexes $\MM = \Ch(k)$.
Since by Theorem \ref{theo:Chmodelcat} we know that
$\Ch(k)$ is a symmetric monoidal model category, it is natural 
to ask whether the category $\Op_\CCC(\Ch(k))$ of $\CCC$-colored 
operads with values in chain complexes carries an induced model structure. 
In the most general scenario, i.e.\ colored operads with values in {\em any} 
symmetric monoidal model category, this would not be the case. However, there
has been a large amount of research on establishing criteria which 
ensure that the category of operads (as well as the category of algebras 
over an operad, see Section \ref{subsec:algebras} below) 
carries a canonical model structure induced by the free-forget adjunction, cf.\
\cite{HinichOriginal,Spitzweck,BergerMoerdijkOriginal,BergerMoerdijk,CisinskiMoerdijk,PavlovScholbach,Hinich}.
The case of interest for us, i.e.\ chain complexes $\Ch(k)$ of $k$-modules
with $k\supseteq \bbQ$ (this technical condition is important here), has been understood
by Hinich \cite{HinichOriginal,Hinich}, who has proven the following result.
\begin{theo}\label{theo:Opmodelcat}
Let $\CCC\in\Set$ be any non-empty set of colors.
Define a morphism $\phi : \O\to \P$ in $\Op_\CCC(\Ch(k))$ to be
\begin{itemize}
\item[(i)] a weak equivalence if each component 
$\phi : \O\big(\substack{t\\ \und{c}}\big) \to \P\big(\substack{t \\ \und{c}}\big)$
is a weak equivalence in $\Ch(k)$, i.e.\ a quasi-isomorphism;

\item[(ii)] a fibration if each component
$\phi : \O\big(\substack{t\\ \und{c}}\big) \to \P\big(\substack{t \\ \und{c}}\big)$
is a fibration in $\Ch(k)$, i.e.\ degree-wise surjective;

\item[(iii)] a cofibration if it has the left lifting property
with respect to all acyclic fibrations.
\end{itemize}
If $k\supseteq \bbQ$, these choices endow $\Op_\CCC(\Ch(k))$
with the structure of a model category.
\end{theo}

\subsection{\label{subsec:algebras}Homotopy theory of algebras over colored operads}
An algebra over a colored operad $\O$ should be interpreted
as a concrete realization/representation of the abstract operations encoded in
$\O$. Explicitly, an {\em algebra $A$ over a $\CCC$-colored operad} $\O\in \Op_\CCC(\MM)$
with values in $\MM$ (also called an {\em $\O$-algebra}) is given by the following data:
\begin{itemize}
\item[(i)] for each color $c\in\CCC$, an object
\begin{flalign}
A_c \in\MM\quad;
\end{flalign}

\item[(ii)] for each $n\geq 0$ and $n+1$-tuple of colors 
$(\und{c},t) \in\CCC^{n+1}$, an $\MM$-morphism (called {\em $\O$-action})
\begin{flalign}
\alpha  :  \O\big(\substack{t\\ \und{c}}\big) \otimes A_{\und{c}} \longrightarrow A_t\quad,
\end{flalign}
where $A_{\und{c}} := \bigotimes_{i=1}^n A_{c_i}$ with the usual 
convention that $A_{\emptyset}=I$ for the empty tuple $\emptyset \in \CCC^0$.
\end{itemize}
These data are required to satisfy the following compatibility conditions:
1.)~The $\O$-action is equivariant under permutations, i.e.\ the diagrams
\begin{flalign}\label{eqn:alg1}
\xymatrix{
\ar[rr]^-{\alpha} \O\big(\substack{t\\ \und{c}}\big) \otimes A_{\und{c}} \ar[d]_-{\O(\sigma) \otimes \text{permute}}
&& A_t\\
\O\big(\substack{t\\ \und{c}\sigma}\big) \otimes A_{\und{c}\sigma} \ar[rru]_-{\alpha}&& 
}
\end{flalign}
commute. 2.)~The $\O$-action is compatible with the operadic composition $\gamma$, i.e.\ the diagrams
\begin{flalign}\label{eqn:alg2}
\xymatrix@C=2em{
\Big(\O\big(\substack{t\\\und{a}}\big) \otimes \bigotimes\limits_{i=1}^n  \O\big(\substack{a_i\\\und{b}_i}\big)\Big) \otimes
A_{\und{b}} \ar[d]_{\text{permute}}\ar[rrrr]^-{\gamma\otimes\id} &&&& \O\big(\substack{t \\ \und{b}}\big) \otimes
A_{\und{b}}\ar[d]^-{\alpha}\\
\O\big(\substack{t\\\und{a}}\big) \otimes\bigotimes\limits_{i=1}^n \Big( \O\big(\substack{a_i\\\und{b}_i}\big) \otimes
A_{\und{b}_i}\Big)  \ar[rr]_-{\id\otimes \Motimes_{i=1}^n\alpha}&&\O\big(\substack{t\\\und{a}}\big) \otimes A_{\und{a}} \ar[rr]_-{\alpha} && A_t
}
\end{flalign}
commute. 3.)~The $\O$-action of the operadic unit $\oone$ is trivial, i.e.\
the diagrams
\begin{flalign}\label{eqn:alg3}
\xymatrix@C=1em{
\ar[drr]_-{\cong} I\otimes A_c\ar[rr]^-{\oone\otimes \id} && \O\big(\substack{c \\ c}\big) \otimes A_c\ar[d]^-{\alpha}\\
&& A_c
}
\end{flalign}
commute. The $\O$-algebras form a category, which we denote by $\Alg(\O)$.
Concretely, a morphism $\kappa : A \to B$ between two $\O$-algebras
is a collection of $\MM$-morphisms
\begin{flalign}
\kappa :  A_c \longrightarrow B_c\quad,
\end{flalign}
for all colors $c\in\CCC$, that is compatible with the $\O$-actions, 
i.e.\ $\kappa \, \alpha^A = \alpha^B\,\big(\id\otimes \Motimes_{i=1}^n\kappa\big)$.
\sk

Now let us focus on the case of chain complexes $\MM = \Ch(k)$.
Similarly to the scenario in Section \ref{subsec:operads},
it is in general not true that the category of $\O$-algebras for
any colored operad $\O$ with values in any symmetric monoidal 
model category carries a canonical model structure induced
by the free-forget adjunction. If it does, the colored operad
$\O$ is called {\em admissible} in the standard terminology of
\cite{HinichOriginal,Spitzweck,BergerMoerdijkOriginal,BergerMoerdijk,CisinskiMoerdijk,PavlovScholbach,Hinich}.
The case of interest for us, i.e.\ chain complexes $\Ch(k)$ of $k$-modules
with $k\supseteq \bbQ$ (this technical condition is again important), has been understood
by Hinich \cite{HinichOriginal,Hinich}, who has proven the following result.
\begin{theo}\label{theo:Algmodelcat}
Let $\CCC\in \Set$ be any non-empty set of colors
and $\O\in\Op_\CCC(\Ch(k))$ any $\CCC$-colored
operad with values in chain complexes. Define a
morphism $\kappa : A \to B$ in $\Alg(\O)$ to be
\begin{itemize}
\item[(i)] a weak equivalence if each component 
$\kappa : A_c \to B_c$
is a weak equivalence in $\Ch(k)$, i.e.\ a quasi-isomorphism;

\item[(ii)] a fibration if each component
$\kappa : A_c \to B_c$
is a fibration in $\Ch(k)$, i.e.\ degree-wise surjective;

\item[(iii)] a cofibration if it has the left lifting property
with respect to all acyclic fibrations.
\end{itemize}
If $k\supseteq \bbQ$, these choices endow $\Alg(\O)$ with 
the structure of a model category. In other words, each $\CCC$-colored operad
$\O\in \Op_\CCC(\Ch(k))$ is admissible when $k\supseteq \bbQ$.
\end{theo}

Before we can conclude this section, we still have to discuss
the compatibility (in a sense to be made precise below) of these 
model structures under a change of operad, for example by
an $\Op_\CCC(\Ch(k))$-morphism. For our paper we shall
need a more flexible variant of changing the operad which does not
necessarily preserve the set of colors. This is captured by 
the following definition.
\begin{defi}\label{def:variablecolorsoperads}
The {\em category $\Op(\Ch(k))$ of operads with varying colors} is defined as follows:
An object is a pair $(\CCC,\O)$ consisting of a non-empty set
$\CCC\in\Set$ and a $\CCC$-colored operad $\O\in\Op_\CCC(\Ch(k))$.
A morphism is a pair $(f,\phi) : (\CCC,\O)\to(\DDD,\P)$
consisting of a map of sets $f : \CCC\to\DDD$
and an $\Op_\CCC(\Ch(k))$-morphism $\phi : \O\to f^\ast(\P)$
to the pullback of $\P$ along $f$. The latter is defined by the 
components $f^\ast(\P)\big(\substack{t\\ \und{c}}\big) := 
\P\big(\substack{f(t) \\ f(\und{c})}\big)$ and the obvious 
restriction of the operad structure on $\P$.
\end{defi}

For every $\Op(\Ch(k))$-morphism $(f,\phi) : (\CCC,\O)\to (\DDD,\P)$,
we can define a pullback functor $(f,\phi)^\ast : \Alg(\P)\to \Alg(\O)$
between the categories of algebras. Concretely, given any $\P$-algebra
$A$, i.e.\ a collection of chain complexes $A_d \in \Ch(k)$, for all
$d\in\DDD$, equipped with a $\P$-action $\alpha : \P\big(\substack{s \\ \und{d} }\big)\otimes A_{\und{d}} \to A_s$,
we define the $\O$-algebra $(f,\phi)^\ast A$ by the collection of chain complexes
$((f,\phi)^\ast A)_c := A_{f(c)}$, for all $c\in\CCC$, equipped with the $\O$-action
\begin{flalign}
\xymatrix{
\O\big(\substack{t \\ \und{c}}\big)\otimes ((f,\phi)^\ast A)_{\und{c}}\ar[r]^-{\phi \otimes \id}&
\P\big(\substack{f(t) \\ f(\und{c}) } \big)\otimes A_{f(\und{c})} \ar[r]^-{\alpha} & A_{f(t)}= ((f,\phi)^\ast A)_t
}\quad.
\end{flalign}
It is well-known (see e.g.\ \cite{BeniniSchenkelWoike}) 
that the pullback functor admits a left adjoint, 
i.e.\ we have an adjunction
\begin{flalign}\label{eqn:generaladjunction}
\xymatrix{
(f,\phi)_! \,:\, \Alg(\O) ~\ar@<0.5ex>[r]&\ar@<0.5ex>[l]  ~\Alg(\P) \,:\, (f,\phi)^\ast
}\quad,
\end{flalign}
for every $\Op(\Ch(k))$-morphism $(f,\phi) : (\CCC,\O)\to (\DDD,\P)$.
The left adjoint $(f,\phi)_!$ is often called {\em operadic left Kan extension}.
It is easy to prove that these adjunctions are compatible (in the sense of Quillen adjunctions)
with the model structures given in Theorem \ref{theo:Algmodelcat}.
\begin{propo}\label{propo:Quillenadjunctionalgebras}
For every $\Op(\Ch(k))$-morphism $(f,\phi) : (\CCC,\O)\to (\DDD,\P)$, 
the adjunction \eqref{eqn:generaladjunction} is a Quillen adjunction,
i.e.\ the right adjoint functor $(f,\phi)^\ast$ preserves fibrations and acyclic fibrations.
Furthermore, $(f,\phi)^\ast$ preserves also weak equivalences.
\end{propo}
\begin{proof}
Recall that an $\Alg(\P)$-morphism $\kappa : A\to B$ is a collection
of $\Ch(k)$-morphisms $\kappa : A_d\to B_d$, for all $d\in\DDD$,
that is compatible with the $\P$-actions. Applying the pullback functor
defines an $\Alg(\O)$-morphism
$(f,\phi)^\ast\kappa : (f,\phi)^\ast A\to (f,\phi)^\ast B$,
whose underlying collection of $\Ch(k)$-morphisms
is $\kappa : A_{f(c)}\to B_{f(c)}$, for all $c\in\CCC$.
Because fibrations and weak equivalences are 
defined component-wise in Theorem \ref{theo:Algmodelcat}, 
and hence so are acyclic fibrations, it follows that
the pullback functor preserves these classes of morphisms.
\end{proof}

It remains to clarify under which conditions a weak equivalence
between colored operads induces a Quillen equivalence 
between the corresponding model categories of algebras, which is
the appropriate notion of equivalence between model categories \cite{Dwyer,Hovey}. For this we shall consider
a natural generalization of the notion of weak equivalence 
between $\CCC$-colored operads established in Theorem \ref{theo:Opmodelcat}
to the case of operads with varying colors, cf.\ \cite[Definition 2.4.2]{Hinich}.
In the following definition we denote by $\O_1$ the underlying $\Ch(k)$-enriched
category of a $\CCC$-colored operad $\O \in\Op_\CCC(\Ch(k))$.
Explicitly, the objects in this category are the colors $c\in \CCC$, 
the chain complex of morphisms from $c$ to $t$
is $\O\big(\substack{ t \\ c}\big)\in\Ch(k)$, the composition
of morphisms is via the operadic composition and the identity morphisms
are given by the operadic unit. We further denote by
$H_0(\O_1)$ the (ordinary) category which is obtained by
taking the $0$th homology of all chain complexes of morphisms.
\begin{defi}\label{def:variablecolorWE}
A morphism $(f,\phi) : (\CCC,\O)\to (\DDD,\P)$ in the category of operads with varying colors
$\Op(\Ch(k))$ is called a {\em weak equivalence} if
\begin{itemize}
\item[(1)] each component $\phi : \O\big(\substack{t \\ \und{c}}\big) \to \P\big(\substack{ f(t) \\ f(\und{c}) }\big)$
is a weak equivalence in $\Ch(k)$, i.e.\ a quasi-isomorphism;

\item[(2)] the induced functor $H_0(f,\phi) : H_0(\O_1)\to H_0(\P_1)$ is an equivalence of categories.
\end{itemize}
\end{defi}
\begin{rem}
Notice that a morphism of $\CCC$-colored operads $\phi : \O\to \P$ can be equivalently regarded as an
$\Op(\Ch(k))$-morphism $(\id_\CCC , \phi) : (\CCC,\O)\to (\CCC,\P)$. In this special case, the first condition
in Definition \ref{def:variablecolorWE} implies the second one. This means that the concept
of weak equivalence introduced in Definition \ref{def:variablecolorWE} agrees with the previous
one from Theorem \ref{theo:Opmodelcat} for morphisms between operads with the same underlying set of colors.
\end{rem}

An important class of operads for which such a weak equivalence induces a Quillen equivalence
between their model categories of algebras is given by $\Sigma$-cofibrant colored operads,
cf.\ \cite[Definition 2.4.3]{Hinich}. Let us briefly recall this definition: Given any $\CCC$-colored
operad $\O\in\Op_\CCC(\Ch(k))$, note that each of its components
$\O\big(\substack{t \\ \und{c}}\big) \in\Ch(k)$ carries a canonical right action 
of the group of automorphisms of the profile $\und{c} = (c_1,\dots,c_n)$. Concretely,
this is given by restricting \eqref{eqn:operadequivariances} to the subgroup 
$\Sigma_{\und{c}} \subseteq \Sigma_n$ of those permutations $\sigma$ that preserve 
the profile, i.e.\ $\und{c}\sigma = \und{c}$. As a consequence, one may regard
$\O\big(\substack{t \\ \und{c}}\big)$ as an object in the category $\Ch(k)^{\Sigma_{\und{c}}}$ 
of chain complexes with right $\Sigma_{\und{c}}$-action. (The morphisms
in $\Ch(k)^{\Sigma_{\und{c}}}$ are $\Sigma_{\und{c}}$-equivariant chain complex morphisms.)
We note that the category $\Ch(k)^{\Sigma_{\und{c}}}$ carries  a canonical model
structure in which a morphism is a weak equivalence (respectively a fibration) if, when forgetting
the group actions, the underlying chain complex morphism is a weak equivalence (respectively a fibration)
in $\Ch(k)$. Cofibrations are determined by the left lifting property with respect to all acyclic fibrations.
\begin{defi}\label{def:Sigmacofibrant}
A $\CCC$-colored operad $\O\in \Op_\CCC(\Ch(k))$ is called {\em $\Sigma$-cofibrant}
if each component $\O\big(\substack{t\\ \und{c}} \big)$ is a cofibrant object in the model
category $\Ch(k)^{\Sigma_{\und{c}}}$. 
\end{defi}

The following result has been proven by Hinich \cite{Hinich,HinichOriginal}.
\begin{theo}\label{theo:Quillenequivalence}
Consider a weak equivalence $(f,\phi) : (\CCC,\O)\to (\DDD,\P)$ in the sense 
of Definition \ref{def:variablecolorWE}. If both $\O$ and $\P$ are $\Sigma$-cofibrant, 
then the Quillen adjunction in \eqref{eqn:generaladjunction} is a Quillen equivalence.
\end{theo}

In our case of interest $k\supseteq\bbQ$, we have the following
straightforward result which implies that every colored operad whose components are 
cofibrant chain complexes is $\Sigma$-cofibrant.
\begin{propo}\label{propo:componentcofimpliesSigmacof}
Let  $k\supseteq \bbQ$ and $G$ be any finite group.
Then every object $X\in \Ch(k)^G$ whose underlying chain 
complex is a cofibrant object in $\Ch(k)$ is a cofibrant object in $\Ch(k)^G$. 
As a consequence, every $\CCC$-colored operad $\O\in\Op_\CCC(\Ch(k))$
whose components $\O\big(\substack{t\\\und{c}}\big)$ are cofibrant objects
in $\Ch(k)$ is $\Sigma$-cofibrant.
\end{propo}
\begin{proof}
By definition of cofibrations in $\Ch(k)^G$, 
we have to prove that there exists a lifting
\begin{flalign}\label{eqn:liftingproblem}
\xymatrix{
&& Y\ar[d]^-{p}\\
X\ar@{-->}[rru]^-{l} \ar[rr]_-{h}&& Z
}
\end{flalign}
in the category $\Ch(k)^G$, for all acyclic fibrations $p : Y\to Z$ and all morphisms $h : X\to Z$ in $\Ch(k)^G$.
Forgetting the right $G$-actions in this diagram, we obtain a lifting $\widetilde{l} : X\to Y$ in the category
of chain complexes  because $X$ is by hypothesis a cofibrant object in $\Ch(k)$.
However, such $\widetilde{l}$ is not necessarily $G$-equivariant.
Using the crucial property that $k\supseteq \bbQ$ contains the rationals, we can always define a
$G$-equivariant chain complex morphism by group averaging
\begin{flalign}
l :  X\longrightarrow Y~,~~x\longmapsto l(x) := \frac{1}{\vert G\vert} \sum_{g\in G} \big(\widetilde{l}(x \cdot g) \big)\cdot g^{-1}\quad,
\end{flalign}
where we denote right $G$-actions by $\cdot$ and $\vert G\vert$ is the number of elements in $G$.
It is easy to confirm that this defines a lifting in \eqref{eqn:liftingproblem}.
\end{proof}


\section{\label{sec:strict}Strict algebraic quantum field theories in chain complexes}

\subsection{Orthogonal categories and quantum field theories}
Given any small category $\CC$, we denote by
$\mathrm{Mor}\,\CC \, {}_\mathrm{t}^{}\!\times\!{}^{}_\mathrm{t} \,\mathrm{Mor}\,\CC$
the set of pairs of $\CC$-morphisms whose targets coincide. An element
in this set is of the form $(f_1 : c_1 \to c, f_2 : c_2\to c)$. The following
concept of orthogonal categories has been introduced in \cite{BeniniSchenkelWoike}.
\begin{defi}\label{def:OCat}
\begin{itemize}
\item[a)] An {\em orthogonal category} is a pair $\ovr{\CC} = (\CC,\perp)$
consisting of a small category $\CC$ and a subset ${\perp} \subseteq 
\mathrm{Mor}\,\CC \, {}_\mathrm{t}^{}\!\times\!{}^{}_\mathrm{t} \,\mathrm{Mor}\,\CC$
(called {\em orthogonality relation}) satisfying the following properties:
\begin{itemize}
\item[(1)] {\em Symmetry:} If $(f_1,f_2)\in {\perp}$, then $(f_2,f_1)\in{\perp}$.
\item[(2)] {\em Stability under post-composition:} If $(f_1,f_2)\in{\perp}$, then
$(g\,f_1, g\, f_2)\in{\perp}$, for all composable $\CC$-morphisms $g$.
\item[(3)] {\em Stability under pre-composition:} If $(f_1,f_2)\in{\perp}$, then
$(f_1\,h_1,  f_2\,h_2)\in{\perp}$, for all composable $\CC$-morphisms $h_1$ and $h_2$.
\end{itemize}
Elements $(f_1,f_2)\in{\perp}$ are called {\em orthogonal pairs} and they are also denoted by $f_1\perp f_2$.

\item[b)] Given two orthogonal categories $\ovr{\CC}=(\CC,\perp_\CC)$ and 
$\ovr{\DD} = (\DD,\perp_\DD)$, an {\em orthogonal functor} $F : \ovr{\CC}\to \ovr{\DD}$ 
is a functor $F : \CC\to\DD$ that preserves the orthogonality relations, i.e.\
such that $f_1\perp_\CC f_2$ implies $F(f_1) \perp_\DD F(f_2)$. 

\item[c)] We denote by $\OCat$ the category whose objects are orthogonal categories and whose morphisms
are orthogonal functors.
\end{itemize}
\end{defi}

Now we shall present examples of orthogonal categories that are relevant
for locally covariant quantum field theory \cite{Brunetti,FewsterVerch}. 
Further examples can be found in \cite{BeniniSchenkelWoike}, 
including orthogonal categories that are relevant for 
ordinary algebraic quantum field theories on a fixed spacetime and
chiral conformal quantum field theories.
\begin{ex}\label{ex:LocOCat}
Let $\Loc$ be any small category that is equivalent to the usual category of oriented, 
time-oriented and globally hyperbolic Lorentzian manifolds of a fixed dimension $m \geq 2$, 
see e.g.\ \cite{Brunetti,FewsterVerch}. We often follow the standard 
terminology and use the word {\em spacetimes} for objects $M\in\Loc$  
and {\em spacetime embeddings} for $\Loc$-morphisms $f : M\to M^\prime$. 
Let us equip $\Loc$ with the following orthogonality relation:
$(f_1 : M_1\to M)\perp_\Loc (f_2 : M_2\to M)$ if and only if the images
$f_1(M_1)$ and $f_2(M_2)$ are causally disjoint subsets in $M$.
Then $\ovr{\Loc} := (\Loc,\perp_\Loc) \in \OCat$ is an orthogonal category.
\sk

Let $\Locc\subseteq \Loc$ be the full subcategory of all oriented, 
time-oriented and globally hyperbolic Lorentzian manifolds whose underlying manifold
is diffeomorphic to $\bbR^m$. With a slight abuse of notation, we shall refer to objects $M\in\Locc$ 
as {\em diamond spacetimes}. The orthogonality relation 
on $\ovr{\Loc}$ restricts to an orthogonality relation on $\Locc\subseteq \Loc$. Explicitly,
 $(f_1 : M_1\to M)\perp_{\Locc} (f_2 : M_2\to M)$
if and only if $(f_1 : M_1\to M)\perp_\Loc (f_2 : M_2\to M)$. The embedding
functor $j: \Locc \to\Loc$ defines an orthogonal functor $j : \ovr{\Locc}\to\ovr{\Loc}$.
\end{ex}

Let us denote by $\dgAlg(k)$ the category of (possibly unbounded) {\em differential graded
algebras} over a commutative and unital ring $k\supseteq \bbQ$, i.e.\ $\dgAlg(k)$
is the category of associative and unital algebras in the symmetric monoidal 
category $\Ch(k)$ of chain complexes of $k$-modules, cf.\ Section \ref{subsec:complexes}. 
Given any small category $\CC$, we consider the functor category $\dgAlg(k)^{\CC}$
whose objects are covariant functors $\AAA : \CC\to \dgAlg(k)$
and whose morphisms are natural transformations $\zeta : \AAA \Rightarrow \BBB$
between such functors $\AAA,\BBB : \CC\to\dgAlg(k)$. 
Algebraic quantum field theories with values in chain complexes
may be formalized by the following general definition.
\begin{defi}\label{def:QFTcats}
Let $\ovr{\CC} = (\CC,\perp)$ be an orthogonal category. 
\begin{itemize}
\item[a)] A functor $\AAA: \CC\to\dgAlg(k)$ is called
{\em $\perp$-commutative} if for all  $(f_1: c_1\to c)\perp (f_2:c_2\to c)$
the diagram
\begin{flalign}
\xymatrix{
\AAA(c_1) \otimes \AAA(c_2) \ar[rr]^-{\AAA(f_1)\otimes \AAA(f_2)} \ar[d]_-{\AAA(f_1)\otimes \AAA(f_2)} && \AAA(c)\otimes \AAA(c) \ar[d]^-{\mu_c^\op}\\
\AAA(c)\otimes \AAA(c) \ar[rr]_-{\mu_c}&& \AAA(c)
}
\end{flalign}
in $\Ch(k)$ commutes. Here $\mu_c$ (respectively $\mu_c^\op :=\mu_c\,\tau $) denotes the (opposite)
multiplication on the differential graded algebra $\AAA(c)\in\dgAlg(k)$.

\item[b)] The category of {\em $\Ch(k)$-valued quantum field theories on $\ovr{\CC}$}
is defined as the full subcategory
\begin{flalign}
\QFT(\ovr{\CC}) \,\subseteq\, \dgAlg(k)^\CC
\end{flalign}
whose objects are all $\perp$-commutative functors. 
\end{itemize}
\end{defi}

\begin{ex}\label{ex:LCQFT}
For the orthogonal category $\ovr{\Loc}$ defined in Example \ref{ex:LocOCat},
a functor $\AAA : \Loc\to\dgAlg(k)$ is $\perp$-commutative precisely 
when the restricted graded commutators
\begin{flalign}
\big[\AAA(f_1)(-), \AAA(f_2)(-)\big] \,:\, \AAA(M_1)\otimes\AAA(M_2)~\longrightarrow~ \AAA(M)
\end{flalign}
are zero, for all pairs of spacetime embeddings $f_1 : M_1\to M$ and $f_2 : M_2\to M$ 
whose images are causally disjoint. Hence, the $\perp$-commutativity property 
of the functor $\AAA$ is precisely the (graded) Einstein causality axiom of locally 
covariant quantum field theory. The category $\QFT(\ovr{\Loc})$ from
Definition \ref{def:QFTcats} is thus the category of $\Ch(k)$-valued {\em off-shell} 
locally covariant quantum field theories, i.e.\ functors $\AAA : \Loc\to\dgAlg(k)$ 
that satisfy the Einstein causality axiom, but not necessarily the time-slice axiom. 
If we would like to focus only on those theories which satisfy time-slice, we can 
localize the orthogonal category $\ovr{\Loc}$ 
at the set of Cauchy morphism $W$, which results in another orthogonal category 
$\ovr{\Loc[W^{-1}]}$ together with an orthogonal localization functor
$L : \ovr{\Loc} \to \ovr{\Loc[W^{-1}]}$, see \cite{BeniniSchenkelWoike}. 
The category $\QFT(\ovr{\Loc[W^{-1}]})$ then describes $\Ch(k)$-valued 
{\em on-shell} locally covariant quantum field theories, i.e.\  functors 
$\AAA : \Loc\to\dgAlg(k)$ that satisfy both the Einstein causality axiom and 
the time-slice axiom. 
\end{ex}

\begin{ex}\label{ex:BRSTBV}
Constructions of perturbative quantum gauge theories via the BRST/BV formalism
produce $\perp$-commutative functors $\AAA : \Loc\to \dgAlg(k)$, where
$k = \bbC[[\hbar,g]]$ is the ring of formal power series in Planck's constant 
$\hbar$ and the coupling constant $g$, see e.g.\ \cite{Hollands,FredenhagenRejzner,FredenhagenRejzner2}. 
Here unbounded chain complexes are crucial to support both the ghost fields (having negative
degree in our homological degree convention) and the anti fields (having positive degree).
\end{ex}

\subsection{The underlying operads}
We have shown in \cite{BeniniSchenkelWoike} that the category
$\QFT(\ovr{\CC})$ of quantum field theories on an orthogonal category
$\ovr{\CC}$ (cf.\ Definition \ref{def:QFTcats}) admits a description in terms of 
the category of algebras over a suitable colored operad $\O_{\ovr{\CC}}$. 
Before we can provide a definition of this operad, we have to introduce some notation.
For a small category $\CC$, we denote by $\CC_0$ its set of objects.
Given any tuple $(\und{c},t) = ((c_1,\dots,c_n),t)\in\CC_0^{n+1}$ of objects,
we denote by $\CC(\und{c},t) := \prod_{i=1}^n \CC(c_i,t)$ the product of
$\Hom$-sets. Its elements will be denoted by symbols like
$\und{f} = (f_1,\dots,f_n)\in \CC(\und{c},t)$. For a set
$S\in\Set$ and a chain complex $V\in\Ch(k)$, we define the tensoring
\begin{flalign}\label{eqn:tensoringgeneral}
S\otimes V \,:= \, \bigoplus_{s\in S} V\, \in \, \Ch(k)\quad.
\end{flalign}
For example, the tensoring
\begin{flalign}\label{eqn:tensoring}
S\otimes k = \bigoplus_{s\in S} k\, \in \, \Ch(k)
\end{flalign}
of a set $S$ and the monoidal unit  $k\in\Ch(k)$
is a chain complex concentrated in degree $0$
with trivial differential $\dd=0$. The $k$-module
in degree $0$ is the free $k$-module generated by $S$.
With a slight abuse of notation, we denote
for an element $s\in S$ the unit element $1\in k$
of the $s$-component of the direct sum \eqref{eqn:tensoring} 
by the same symbol $s\in S\otimes k$.
\begin{defi}\label{def:AQFToperads}
Let $\ovr{\CC} = (\CC,\perp)$ be an orthogonal category.
The {\em $\CC_0$-colored operad $\O_{\ovr{\CC}}\in\Op_{\CC_0}(\Ch(k))$ 
of $\Ch(k)$-valued quantum field theories on $\ovr{\CC}$}
is defined by the following data:
\begin{itemize}
\item[(i)] For any $(\und{c},t) \in\CC_0^{n+1}$, we define the chain complex
of operations by
\begin{flalign}\label{eqn:AQFToperadoperations}
\O_{\ovr{\CC}}\big(\substack{t \\ \und{c}}\big) \,:=\, \big(\Sigma_n \times \CC(\und{c},t)\big)\big/{\sim_\perp}\otimes k\,\in\,\Ch(k)\quad,
\end{flalign}
where $\Sigma_n$ is the symmetric group on $n$ letters and 
the equivalence relation is defined as follows:
$(\sigma,\und{f}) \sim_\perp (\sigma^\prime,\und{f}^\prime)$ if and only if
$\und{f} = \und{f}^\prime$ and the right permutation
$\sigma\sigma^{\prime\,-1} : \und{f}\sigma^{-1} \to \und{f}\sigma^{\prime\,-1}$
is generated by transpositions of adjacent orthogonal pairs.

\item[(ii)] For any $c\in \CC_0$, we define the operadic unit by 
\begin{flalign}
\oone : k \longrightarrow \O_{\ovr{\CC}}\big(\substack{c \\ c}\big)~,~~1 \longmapsto
[e,\id_c]\quad,
\end{flalign}
where $e\in\Sigma_1$ is the identity permutation.

\item[(iii)] For any $(\und{a},t)\in \CC_0^{n+1}$ and $(\und{b}_i,a_i)\in \CC_0^{k_i +1}$, 
we define the operadic composition by
\begin{flalign}
\nn \gamma :   \O_{\ovr{\CC}}\big(\substack{t\\\und{a}}\big) \otimes\bigotimes_{i=1}^n 
\O_{\ovr{\CC}}\big(\substack{a_i\\\und{b}_i}\big) & ~\longrightarrow~
\O_{\ovr{\CC}}\big(\substack{t\\ \und{b}}\big)~~,\\
[\sigma,\und{f}]\otimes \Motimes_{i=1}^n [\sigma_i,\und{g}_i]
&~ \longmapsto~ \big[ \sigma(\sigma_1,\dots,\sigma_n) , \und{f}(\und{g}_1,\dots,\und{g}_n)\big]\quad,
\end{flalign}
where $\sigma(\sigma_1,\dots,\sigma_n) = \sigma\langle k _{\sigma^{-1}(1)} , \dots, k_{\sigma^{-1}(n)}\rangle
\, (\sigma_1\oplus\cdots\oplus \sigma_n)$ is the group multiplication in $\Sigma_{k_1+\cdots+k_n}$
of the corresponding block permutation and block sum permutation, and 
$\und{f}(\und{g}_1,\dots,\und{g}_n) = (f_1\, g_{11}, \dots, f_n\, g_{n k_n})$ is given by 
composition of $\CC$-morphisms.

\item[(iv)] For any $(\und{c},t) \in\CC_0^{n+1}$ and $\sigma^\prime\in\Sigma_n$,  
we define the permutation action by
\begin{flalign}
\O_{\ovr{\CC}}(\sigma^\prime) :  \O_{\ovr{\CC}}\big(\substack{t\\ \und{c}}\big)\longrightarrow
\O_{\ovr{\CC}}\big(\substack{t\\ \und{c}\sigma^\prime}\big)~,~~[\sigma,\und{f}] \longmapsto
[\sigma\sigma^\prime, \und{f}\sigma^\prime]\quad.
\end{flalign}
\end{itemize}
\end{defi}

\noindent The following results have been proven in \cite{BeniniSchenkelWoike}.
\begin{propo}\label{propo:functorialityoperads}
The assignment $\ovr{\CC}\mapsto \O_{\ovr{\CC}} $ of the colored operads
from Definition \ref{def:AQFToperads} naturally extends to a functor $\O_{(-)} : \OCat\to \Op(\Ch(k))$.
\end{propo}

\begin{theo}\label{theo:QFTsareOalgebras}
For any orthogonal category $\ovr{\CC}$ there exists an isomorphism
\begin{flalign}
\Alg(\O_{\ovr{\CC}}) \,\cong\, \QFT(\ovr{\CC})
\end{flalign}
between the category of $\O_{\ovr{\CC}}$-algebras and the category of
$\Ch(k)$-valued quantum field theories on $\ovr{\CC}$, cf.\ Definition \ref{def:QFTcats}.
This isomorphism is natural in $\ovr{\CC}\in\OCat$.
\end{theo}

\begin{ex}\label{ex:emphasizeOalgebraperspective}
Consider the orthogonal category $\ovr{\Loc}$ defined in Example \ref{ex:LocOCat}.
Then algebras over the colored operad $\O_{\ovr{\Loc}}$ are precisely $\Ch(k)$-valued
off-shell locally covariant quantum field theories, cf.\ Example \ref{ex:LCQFT}. Furthermore,
algebras over the colored operad $\O_{\ovr{\Loc[W^{-1}]}}$ corresponding to the
orthogonal category localized at Cauchy morphisms 
are precisely $\Ch(k)$-valued on-shell locally covariant quantum field theories. 
\end{ex}

\subsection{\label{subsec:homotopytheoryQFT}Homotopy theory}
We are now in the position to endow the category $\QFT(\ovr{\CC})$
of $\Ch(k)$-valued quantum field theories on an orthogonal category 
$\ovr{\CC}$ with a model structure. Our choice of model structure
is determined by the canonical model structure on the category 
of algebras over a colored operad in chain complexes (cf.\ Theorem
\ref{theo:Algmodelcat}) and our result in Theorem \ref{theo:QFTsareOalgebras}.
Even though model categories are rather abstract concepts,
we will try to explain by various examples in this section why they are crucial
for many practical constructions with $\Ch(k)$-valued quantum field theories.
\begin{theo}\label{theo:QFTmodelcat}
Let us assume as before that $k\supseteq \bbQ$ contains the rationals. Let $\ovr{\CC}$ be any
orthogonal category and consider the corresponding category  $\QFT(\ovr{\CC})$
of $\Ch(k)$-valued quantum field theories, cf.\ Definition \ref{def:QFTcats}.
Define a morphism $\zeta :\AAA\Rightarrow\BBB$ in $\QFT(\ovr{\CC})$ (i.e.\ a
natural transformation between functors  $\AAA,\BBB: \CC\to\dgAlg(k)$) to be
\begin{itemize}
\item[(i)] a weak equivalence if the underlying $\Ch(k)$-morphism of each
component $\zeta_c : \AAA(c) \to \BBB(c)$ is a quasi-isomorphism;

\item[(ii)] a fibration if the underlying $\Ch(k)$-morphism of each
component $\zeta_c : \AAA(c) \to \BBB(c)$ is degree-wise surjective;

\item[(iii)] a cofibration if it has the left lifting property (cf.\ Remark \ref{rem:lifting}) with respect to
all acyclic fibrations.
\end{itemize}
These choices endow $\QFT(\ovr{\CC})$ with the structure of a model category.
\end{theo}
\begin{proof}
By Theorem \ref{theo:QFTsareOalgebras}, we have an isomorphism
of categories $\QFT(\ovr{\CC})\cong \Alg(\O_{\ovr{\CC}})$. Because
$k\supseteq \bbQ$, the colored operad $\O_{\ovr{\CC}}\in\Op_{\CC_0}(\Ch(k))$
is admissible and hence its category of algebras $\Alg(\O_{\ovr{\CC}})$ carries
the canonical model structure given in Theorem \ref{theo:Algmodelcat}.
The isomorphism $\QFT(\ovr{\CC})\cong \Alg(\O_{\ovr{\CC}})$ induces precisely 
the model structure on $\QFT(\ovr{\CC})$ that was claimed in the statement 
of this theorem.
\end{proof}

\begin{rem}
As a side-remark, we would like to mention that the model 
structure of Theorem \ref{theo:QFTmodelcat} is related to
the projective model structure on functor categories: Consider the
category $\dgAlg(k)^\CC$ of {\em all} functors from $\CC$ to
differential graded algebras. Because $\dgAlg(k) = \Alg(\mathsf{As})$
is the category of algebras over the ($1$-colored, i.e.\ $\CCC=1$ is a singleton) 
associative operad $\mathsf{As}\in \Op_{1}(\Ch(k))$, it carries a canonical model
structure by Theorem \ref{theo:Algmodelcat}. Concretely,
a $\dgAlg(k)$-morphism $\kappa : A\to B$ is a weak equivalence
(respectively a fibration) if its underlying $\Ch(k)$-morphism
is a quasi-isomorphism (respectively degree-wise surjective).
Then we can consider the projective model structure
on the functor category $\dgAlg(k)^\CC$. Concretely, this
means that a natural transformation $\zeta : \AAA\Rightarrow\BBB$
between two functors $\AAA,\BBB : \CC\to\dgAlg(k)$
is a weak equivalence (respectively a fibration)
if the $\Ch(k)$-morphism underlying each component
$\zeta_c : \AAA(c)\to \BBB(c)$ is a quasi-isomorphism (respectively
degree-wise surjective). Restricting the weak equivalences and fibrations
of this model structure
to the full subcategory $\QFT(\ovr{\CC})\subseteq \dgAlg(k)^\CC$
of $\perp$-commutative functors yields the same model structure
as the one in Theorem \ref{theo:QFTmodelcat} above.
It is important to emphasize that this observation {\em does not}
mean that our operadic formalism is irrelevant for endowing
the category of quantum field theories with a model structure.
In fact, it is generically not true that restricting weak equivalences and fibrations
to a full subcategory induces a model structure on it, i.e.\ our operadic
approach to endowing $\QFT(\ovr{\CC})$ with a 
model structure is more intrinsic and fundamental. 
\end{rem}

\begin{ex}\label{ex:BRSTBVweakequivalences}
To be more explicit, let us discuss the result of Theorem \ref{theo:QFTmodelcat} in the context of
the orthogonal category $\ovr{\Loc}$ from Example \ref{ex:LocOCat},
which is relevant for locally covariant quantum field theory, see
also Example \ref{ex:LCQFT}. (We consider off-shell theories in this example.
All statements below are of course also true for on-shell theories
by replacing $\ovr{\Loc}$ with its orthogonal localization $\ovr{\Loc[W^{-1}]}$ at Cauchy morphisms.) 
Theorem \ref{theo:QFTmodelcat} shows that the category $\QFT(\ovr{\Loc})$ 
of all $\Ch(k)$-valued locally covariant quantum field theories
carries a canonical model structure. Recalling the main intuitive 
principle of model category theory, this allows us to consistently regard two theories 
$\AAA : \Loc \to\dgAlg(k)$ and $\BBB : \Loc \to\dgAlg(k)$ as being the same
not only when they are isomorphic, but also when they are weakly equivalent.
A weak equivalence is concretely given by a natural transformation 
$\zeta : \AAA\Rightarrow\BBB$ whose components 
$\zeta_M : \AAA(M)\to \BBB(M)$ are quasi-isomorphisms between
the chain complexes of observables, for all spacetimes $M\in\Loc$. 
At a more informal level, such weak equivalences are frequently used 
in practice for constructing perturbative quantum gauge theories via the 
BRST/BV formalism, see e.g.\ \cite{Hollands,FredenhagenRejzner,FredenhagenRejzner2}. 
For example, the usual technique of adding various auxiliary fields to the differential graded
algebras of observables without changing their homologies should be understood as a particular
instance of our precise concept of weak equivalences developed in 
Theorem \ref{theo:QFTmodelcat}. The theories resulting from adding different choices 
of auxiliary fields are in general not isomorphic, but only weakly equivalent in the model category
$\QFT(\ovr{\CC})$. Model category theory provides a consistent
framework that allows us to treat such theories as if they would be the same.
We comment more on the last point in Examples \ref{ex:derivedFredenhagen} and
\ref{ex:derivedBRSTBV} below.
\end{ex}

Now let us consider an orthogonal functor $F :\ovr{\CC}\to\ovr{\DD}$.
Recalling that by Proposition \ref{propo:functorialityoperads} our 
operads are functorial on the category $\OCat$ of orthogonal categories, 
we obtain an $\Op(\Ch(k))$-morphism $\O_F : \O_{\ovr{\CC}} \to \O_{\ovr{\DD}}$.
Then by Proposition \ref{propo:Quillenadjunctionalgebras} we obtain 
a Quillen adjunction
\begin{flalign}
\xymatrix{
{\O_F}_! \,:\, \Alg(\O_{\ovr{\CC}}) ~\ar@<0.5ex>[r]&\ar@<0.5ex>[l]  ~\Alg(\O_{\ovr{\DD}}) \,:\, \O_F^\ast
}
\end{flalign}
between the model categories of algebras. Using further the natural 
isomorphism of categories $\Alg(\O_{(-)})\cong \QFT(-)$ 
from Theorem \ref{theo:QFTsareOalgebras}, which is in fact 
by construction (cf.\ Theorem \ref{theo:QFTmodelcat}) a natural 
isomorphism of model categories, we immediately obtain the following result.
\begin{propo}\label{propo:QFTQuillenadjunction}
Let us assume as before that $k\supseteq \bbQ$.
For every orthogonal functor $F :\ovr{\CC}\to\ovr{\DD}$ there exists a Quillen
adjunction
\begin{flalign}\label{eqn:QFTQuillenadjunction}
\xymatrix{
F_! \,:\, \QFT(\ovr{\CC}) ~\ar@<0.5ex>[r]&\ar@<0.5ex>[l]  ~\QFT(\ovr{\DD}) \,:\, F^\ast
}
\end{flalign}
between the model categories of $\Ch(k)$-valued quantum field theories on
$\ovr{\CC}$ and $\ovr{\DD}$, cf.\ Theorem \ref{theo:QFTmodelcat}.
The right adjoint functor $F^\ast$ is explicitly given by
pullback along $F$, i.e.\ for a $\perp$-commutative functor 
$\AAA : \DD\to \dgAlg(k)$ we have that $F^\ast (\AAA) := \AAA\,F : \CC\to\dgAlg(k)$
is given by pre-composition. Moreover, the right adjoint functor $F^\ast$ preserves weak equivalences.
\end{propo}

Using standard techniques from model category theory, called {\em derived functors} \cite{Dwyer,Hovey,Riehl}, 
one can modify in a controlled way both the left and the right adjoint functors 
of the Quillen adjunction in \eqref{eqn:QFTQuillenadjunction} in order to obtain 
new functors that preserve weak equivalences. These derived functors 
provide constructions which are homotopically meaningful in the sense that they
are consistent with our main principle that weakly equivalent objects should 
be regarded as being the same. Because by Proposition \ref{propo:QFTQuillenadjunction}
the right adjoint functor $F^\ast : \QFT(\ovr{\DD})\to \QFT(\ovr{\CC})$
already preserves weak equivalences, there is no need to derive it.
Hence, we may simply choose
\begin{flalign}\label{eqn:rightderived}
\bbR F^\ast :=F^\ast :  \QFT(\ovr{\DD})\longrightarrow \QFT(\ovr{\CC})
\end{flalign}
for the right derived functor. However, the left adjoint functor 
$F_! : \QFT(\ovr{\CC})\to\QFT(\ovr{\DD})$ in general does not
preserve weak equivalences and we have to derive it. The standard technique
for defining a left derived functor $\bbL F_! : \QFT(\ovr{\CC})\to\QFT(\ovr{\DD})$,
which {\em does} preserve weak equivalences, is via cofibrant replacements.
Recall that a {\em cofibrant replacement functor} on a model category (here $\QFT(\ovr{\CC})$)
is an endofunctor $Q : \QFT(\ovr{\CC}) \to \QFT(\ovr{\CC})$ such that
$Q(\AAA)$ is a cofibrant object in $\QFT(\ovr{\CC})$ for every $\AAA\in \QFT(\ovr{\CC})$, 
together with a natural weak equivalence $q : Q\Rightarrow \id$. The model category axioms
ensure that a cofibrant replacement functor always exists. Making any choice of cofibrant replacement,
we define a left derived functor by pre-composition
\begin{flalign}\label{eqn:leftderived}
\bbL F_! := F_!\,Q : \QFT(\ovr{\CC})\longrightarrow\QFT(\ovr{\DD})\quad.
\end{flalign}
By construction, the derived functor $\bbL F_!$ preserves weak equivalences, 
see e.g.\ \cite[Lemma 1.1.12]{Hovey}.
In order to illustrate the relevance of these derived functors
for quantum field theory, let us consider the following examples.
\begin{ex}\label{ex:derivedFredenhagen}
Consider the orthogonal functor $j : \ovr{\Locc}\to \ovr{\Loc}$ from Example \ref{ex:LocOCat},
which describes the embedding of the category of diamond spacetimes 
into the category of all spacetimes. By Proposition \ref{propo:QFTQuillenadjunction}, 
we obtain a Quillen adjunction
\begin{flalign}
\xymatrix{
j_! \,:\, \QFT(\ovr{\Locc}) ~\ar@<0.5ex>[r]&\ar@<0.5ex>[l]  ~\QFT(\ovr{\Loc}) \,:\, j^\ast
}\quad.
\end{flalign}
The physical interpretation of the left and right adjoint functors
is as follows: The right adjoint $j^\ast$ is the obvious restriction functor
that restricts theories defined on the category $\Loc$ 
of all spacetimes to the category $\Locc$ of diamond spacetimes.
The left adjoint $j_!$ is a universal extension functor that extends
theories defined only on diamond spacetimes to all spacetimes.
We have shown in \cite{BeniniSchenkelWoike}
that the universal extension functor $j_!$ is an operadic 
refinement of Fredenhagen's universal algebra construction \cite{Fre1,Fre2,Fre3}.
\sk

When working with $\Ch(k)$-valued quantum field theories, 
there exists a non-trivial notion of weak equivalences 
(cf.\ Theorem \ref{theo:QFTmodelcat}) and, of course, we must 
ensure that both the restriction and the extension of quantum field
theories along $j : \ovr{\Locc}\to \ovr{\Loc}$ preserve these weak equivalences.
This is precisely what the derived functors do for us!  
Concretely, we have observed in \eqref{eqn:rightderived} that the right adjoint functor
$j^\ast$ already preserves weak equivalences and hence there is no need to derive it,
i.e.\ we may set $\bbR j^\ast = j^\ast$. In our context this means that the ordinary
restriction $ j^\ast (\AAA) \in \QFT(\ovr{\Locc})$ of quantum field 
theories $\AAA \in \QFT(\ovr{\Loc})$ is a homotopically meaningful construction. 
On the other hand, the universal extension functor $j_!$ is not yet homotopically 
meaningful and it must be derived according to \eqref{eqn:leftderived} by using 
a cofibrant replacement functor, i.e.\ we consider the left derived functor
\begin{flalign}
\bbL j_! := j_!\,Q : \QFT(\ovr{\Locc})\longrightarrow\QFT(\ovr{\Loc})\quad.
\end{flalign}
Let us stress that the derived extension $\bbL j_! (\AAA_{\text{\large $\diamond$}}^{}) \in\QFT(\ovr{\Loc})$
of a theory $\AAA_{\text{\large $\diamond$}}^{} \in \QFT(\ovr{\Locc})$ on diamond 
spacetimes is in general {\em not}
weakly equivalent to its ordinary extension $j_!(\AAA_{\text{\large $\diamond$}}^{}) \in\QFT(\ovr{\Loc})$.
We refer to Appendix \ref{app:extension} for simple toy-models
which explain the relevance of derived extension functors from 
a mathematical and physical perspective.
\end{ex}

\begin{ex}\label{ex:derivedBRSTBV}
In order to emphasize even more our comments in Example
\ref{ex:derivedFredenhagen}, let us consider the following concrete scenario: 
Let us take any perturbative quantum gauge theory
$\AAA_{\text{\large $\diamond$}}^{} \in \QFT(\ovr{\Locc})$ on the full subcategory $\Locc\subseteq \Loc$ 
of diamond spacetimes, see e.g.\ \cite{Hollands,FredenhagenRejzner,FredenhagenRejzner2} 
for concrete examples. (Constructing quantum gauge theories on $\Locc$ is generically easier 
than constructing them on all of $\Loc$ because all bundles over $M\in\Locc$ can be 
trivialized.) Now suppose that we would do a second 
slightly different construction, e.g.\ by using different auxiliary fields and/or gauge fixings, 
and obtain another theory $\AAA_{\text{\large $\diamond$}}^{\prime} \in \QFT(\ovr{\Locc})$ 
that is weakly equivalent to $\AAA_{\text{\large $\diamond$}}^{}$. 
Then, in the spirit of Fredenhagen's universal algebra construction, we would like 
to extend $\AAA_{\text{\large $\diamond$}}^{}$ and $\AAA_{\text{\large $\diamond$}}^{\prime}$
to the category of all spacetimes $\Loc$. If we would use the ordinary (i.e.\ underived) universal
extension functor $j_!$, it is {\em not} guaranteed that $j_! (\AAA_{\text{\large $\diamond$}}^{})$ 
and $j_! (\AAA_{\text{\large $\diamond$}}^{\prime})$ are weakly equivalent theories 
on $\Loc$. (See Appendix \ref{app:extension} for simple toy-models which show that
underived extension functors do not preserve weak equivalences.) 
That is of course unsatisfactory because $\AAA_{\text{\large $\diamond$}}^{}$ 
and $\AAA_{\text{\large $\diamond$}}^{\prime}$ represent the same quantum
gauge theory as they are weakly equivalent.
Using instead the derived universal extension functor $\bbL j_!$, we ensure that 
$\bbL j_! (\AAA_{\text{\large $\diamond$}}^{})$ and $\bbL j_! (\AAA_{\text{\large $\diamond$}}^{\prime})$ 
are weakly equivalent and hence they present the same quantum gauge theory on $\Loc$, 
as expected from the fact that $\AAA_{\text{\large $\diamond$}}^{}$ and 
$\AAA_{\text{\large $\diamond$}}^{\prime}$ present the same quantum gauge theory on $\Locc$.
\end{ex}

\begin{rem}
Let us recall that if $\AAA\in \QFT(\ovr{\CC})$ is a cofibrant object, then the canonical morphism
$\bbL F_! (\AAA)  = F_! \,Q(\AAA) \xrightarrow{\ F_! q_{\AAA}\ } F_!(\AAA) $
is a weak equivalence in $\QFT(\ovr{\DD})$. This means that, for any cofibrant
object $\AAA\in \QFT(\ovr{\CC})$, it is sufficient to compute the much simpler
underived functor $F_!(\AAA)$ in order to obtain the correct homotopy type 
of $\bbL F_! (\AAA)$. It is therefore interesting
to ask whether the typical examples of $\Ch(k)$-valued quantum field theories,
e.g.\ those produced by the BRST/BV formalism, naturally come as 
cofibrant objects in $\QFT(\ovr{\CC})$. (We thank the anonymous referee for asking this question to us.)
This is unfortunately {\em not} the case because, as we explain below, 
the usual canonical commutation relations in quantum (field) theory are incompatible 
with cofibrancy as they define non-semifree differential graded algebras.
\sk

In order to provide a simple argument, let us take $k=\bbC$ to be the field of complex numbers
and consider for the moment the orthogonal category $\ovr{\CC} = (\{\ast\},\emptyset)$ 
given by a point with the empty orthogonality relation. 
Then $\QFT(\ovr{\CC}) = \dgAlg(\bbC)$ is simply the category of 
(non-commutative) differential graded algebras over $\bbC$.  
Let $A\in\dgAlg(\bbC)$ be 
generated by two generators $x$ and $p$ of degree $0$, modulo the 
ideal generated by $\dd x =0 = \dd p$ and 
the canonical commutation relation $[x,p] = i\,1$. (Note that this 
is the algebra for a $1$-dimensional point particle
in quantum mechanics.) We now prove that $A\in\dgAlg(\bbC)$ is {\em not} a cofibrant object
by constructing an acyclic fibration $\kappa : B\to A$ that does not admit a section.
Let $B\in\dgAlg(\bbC)$ be
generated by two generators $x$ and $p$ of degree $0$
and one generator $c$ of degree $1$, 
modulo the ideal generated by the relations $\dd x =0 = \dd p$, 
$x c=0=cx$, $pc=0=cp$, $c^2=0$, $c \,\dd c=0$ and the
modified canonical commutation relation $[x,p] = i\,(1 + \dd c)$.
Taking the differential of these relations we further obtain that
$x\,\dd c = 0 = \dd c \, x$, $p\, \dd c = 0 = \dd c\, p$, $\dd c\, c =0$
and $\dd c\, \dd c=0$. Hence, a generic element of $B$ is of the form
$b = a(x,p) + \alpha \,\dd c + \beta \,c$, where $a(x,p)$ is a polynomial in 
$x$ and $p$, while $\alpha,\beta \in\bbC$ are complex numbers.
We define a $\dgAlg(\bbC)$-morphism $\kappa : B\to A$ by setting
$\kappa(x)=x$, $\kappa(p)= p$ and $\kappa(c)= 0$ and observe that this is
an acyclic fibration in $\dgAlg(\bbC)$. Any section $\sigma : A\to B$ 
of $\kappa$ must be of the form $\sigma(x) = x + \alpha\, \dd c$ and $\sigma(p) = p +\alpha^\prime \,\dd c$,
for some $\alpha,\alpha^\prime \in\bbC$. It follows that this is incompatible
with the commutation relations in $B$ and $A$. Explicitly,
\begin{flalign}
[\sigma(x),\sigma(p)]_B = [x+\alpha \,\dd c, p+ \alpha^\prime\,\dd c]_B= 
[x,p]_B = i\,(1+\dd c) \neq \sigma\big(i\,1\big) = \sigma\big([x,p]_A\big)\quad,
\end{flalign}
where the subscript indicates the algebra in which we take the commutator.
Hence, $A\in\dgAlg(\bbC)$ is not a cofibrant object.
\sk

A similar argument may be applied object-wise 
to any $\AAA\in\QFT(\ovr{\CC})$ that assigns
differential graded algebras involving (graded) canonical commutation relations,
which in particular includes the models constructed in the
BRST/BV formalism \cite{Hollands,FredenhagenRejzner,FredenhagenRejzner2}.
\end{rem}


\section{\label{sec:resolution}Homotopy algebraic quantum field theories}

\subsection{\label{subsec:Sigmaresolution}$\Sigma$-cofibrant resolutions}
Let us first recall the following standard concept of resolutions for colored operads,
see e.g.\ \cite{HinichOriginal,Spitzweck,BergerMoerdijkOriginal,BergerMoerdijk,CisinskiMoerdijk,PavlovScholbach,Hinich}.
\begin{defi}\label{def:Sigmacofibres}
A {\em $\Sigma$-cofibrant resolution} of a 
colored operad $\O\in \Op_{\CCC}(\Ch(k))$ is a $\Sigma$-cofibrant operad
$\O_\infty \in \Op_{\CCC}(\Ch(k))$ (cf.\ Definition \ref{def:Sigmacofibrant})
together with an acyclic fibration
\begin{flalign}
w : \O_\infty \longrightarrow \O
\end{flalign} 
in the model category $ \Op_{\CCC}(\Ch(k))$, cf.\ Theorem \ref{theo:Opmodelcat}.
\end{defi}

The role of $\Sigma$-cofibrant resolutions is that
they provide a suitable framework to study homotopy-coherent
algebraic structures, i.e.\ {\em homotopy algebras
over operads}. For example, $A_\infty$-algebras (see e.g.\ \cite[Section 9]{LodayVallette}) 
are algebras over a $\Sigma$-cofibrant resolution
$A_{\infty}\to \mathsf{As}$ of the associative operad
and $E_\infty$-algebras (see e.g.\ \cite{BergerFresse})
are algebras over a $\Sigma$-cofibrant resolution
$E_{\infty}\to \mathsf{Com}$ of the commutative operad.
As a colored example, homotopy-coherent diagrams (see e.g.\ \cite{BergerMoerdijk})
are algebras over a $\Sigma$-cofibrant resolution
${\mathsf{Diag}_\CC}_{\infty}\to \mathsf{Diag}_\CC$
of the diagram operad over a small category $\CC$. 
\sk

It is important to emphasize the following facts about
existence and ``uniqueness'' (see below for the precise statement) 
of $\Sigma$-cofibrant resolutions: 
1.)~Recalling that $\Op_{\CCC}(\Ch(k))$ is a model category (cf.\ Theorem \ref{theo:Opmodelcat}),
there exists a cofibrant replacement $q : Q(\O)\to \O$ which defines
a particular example of a $\Sigma$-cofibrant resolution, see e.g.\
\cite[Proposition 4.3]{BergerMoerdijkOriginal}. Hence, $\Sigma$-cofibrant resolutions
exist for every colored operad $\O\in \Op_{\CCC}(\Ch(k))$.
2.)~Let us assume that we have two $\Sigma$-cofibrant resolutions
$w : \O_\infty\to \O$ and $w^\prime : \O_\infty^\prime \to \O$
of a colored operad $\O\in \Op_{\CCC}(\Ch(k))$. Taking also a cofibrant replacement
$q : Q(\O)\to \O$, we obtain a commutative diagram
\begin{flalign}
\xymatrix{
\O_\infty^{~} \ar[r]^-{w} & \O &\ar[l]_-{w^\prime} \O_\infty^\prime\\
& Q(\O)\ar[u]_-{q} \ar@{-->}[ru]_-{l^\prime} \ar@{-->}[lu]^-{l}&
}
\end{flalign}
in $\Op_{\CCC}(\Ch(k))$, where the dashed arrows exist by the left lifting property (cf.\ Remark \ref{rem:lifting}) 
because $Q(\O)$ is a cofibrant object and $w$, $w^\prime$ are by definition acyclic fibrations.
The $2$-out-of-$3$ property of weak equivalences further implies
that the dashed arrows are themselves weak equivalences, hence
by Theorem \ref{theo:Quillenequivalence} they induce a zig-zag
of Quillen equivalences
\begin{flalign}
\xymatrix{
\Alg(\O_\infty) \ar@<-0.5ex>[r]_-{l^\ast} ~&~ \Alg(Q(\O)) ~\ar@<0.5ex>[r]^-{l^\prime_!} \ar@<-0.5ex>[l]_-{l_!}&\ar@<0.5ex>[l]^-{{l^\prime}^\ast}  ~\Alg(\O_\infty^{\prime})
}\quad.
\end{flalign}
This shows that the model categories of algebras over different
$\Sigma$-cofibrant resolutions are ``the same'' in the sense that they are
related by a zig-zag of Quillen equivalences. (In particular,
the corresponding homotopy categories are equivalent as categories.)
\sk

Now let us focus on our case of interest, namely the
quantum field theory operads from Definition \ref{def:AQFToperads}.
Inspired by the concept of homotopy algebras over operads, we
start with the following definition.
\begin{defi}
Let $\ovr{\CC}$ be an orthogonal category and 
$w : {\O_{\ovr{\CC}}}_\infty \to \O_{\ovr{\CC}}$ 
a $\Sigma$-cofibrant resolution. The model category of
{\em homotopy quantum field theories} on $\ovr{\CC}$ corresponding to this resolution
is defined as $\QFT_w(\ovr{\CC}) := \Alg({\O_{\ovr{\CC}}}_\infty)$.
\end{defi}

\begin{theo}\label{theo:strictification}
Let us assume as before that $k\supseteq \bbQ$. For every orthogonal
category $\ovr{\CC}$, the $\CC_0$-colored operad $\O_{\ovr{\CC}}\in \Op_{\CC_0}(\Ch(k))$
is $\Sigma$-cofibrant. As a consequence, for every $\Sigma$-cofibrant
resolution $w : {\O_{\ovr{\CC}}}_\infty \to \O_{\ovr{\CC}}$ there exists
a Quillen equivalence
\begin{flalign}\label{eqn:QuillenequivalencehAQFT}
\xymatrix{
w_! \,:\, \QFT_w(\ovr{\CC}) ~\ar@<0.5ex>[r]&\ar@<0.5ex>[l]  ~\QFT(\ovr{\CC}) \,:\, w^\ast
}
\end{flalign}
between the model categories of strict and homotopy 
quantum field theories on $\ovr{\CC}$.
\end{theo}
\begin{proof}
Using Proposition \ref{propo:componentcofimpliesSigmacof},
it is sufficient to prove that each component
$\O_{\ovr{\CC}}(\substack{t \\ \und{c}}) \in\Ch(k)$ is a cofibrant chain complex.
Using the explicit definition in \eqref{eqn:AQFToperadoperations},
we observe that this is the case because the monoidal unit $k\in \Ch(k)$ 
is cofibrant and the $\Set$-tensoring (i.e.\ forming coproducts \eqref{eqn:tensoring})
preserves cofibrant objects. This proves the first part of our claim.
\sk

Concerning the second claim, consider any $\Sigma$-cofibrant resolution
$w : {\O_{\ovr{\CC}}}_\infty \to \O_{\ovr{\CC}}$. Because both the source
and target of this morphism are $\Sigma$-cofibrant, it follows from
Theorem \ref{theo:Quillenequivalence} that  the adjunction \eqref{eqn:QuillenequivalencehAQFT}
is a Quillen equivalence.
\end{proof}

\begin{rem}\label{rem:srtictification}
Theorem \ref{theo:strictification} should be interpreted as a strictification
result for homotopy quantum field theories.
Let us explain this important point in more detail:
Let $w : {\O_{\ovr{\CC}}}_\infty \to \O_{\ovr{\CC}}$ be a $\Sigma$-cofibrant
resolution, e.g.\ the Boardman-Vogt resolution developed
by Yau for our operads \cite{YauQFT}. Algebras over ${\O_{\ovr{\CC}}}_\infty$
generically describe a weaker concept of quantum field theories, where
functoriality, associativity and/or $\perp$-commutativity might just hold 
in a homotopy-coherent fashion. (The extent of this weakening depends on the choice
of $\Sigma$-cofibrant resolution.) Consider any homotopy quantum field theory 
$\AAA_\infty \in \QFT_w(\ovr{\CC})$ and replace it by a 
weakly equivalent cofibrant object $Q(\AAA_\infty) \to \AAA_\infty$. 
(Notice that $Q(\AAA_\infty)$ and $\AAA_\infty$ describe ``the same'' theory 
from our model categorical perspective because they are weakly equivalent.) 
Using that  \eqref{eqn:QuillenequivalencehAQFT} 
is a Quillen equivalence, we obtain from \cite[Proposition 1.3.13]{Hovey} 
that the composite of
\begin{flalign}
\xymatrix@C=4em{
Q (\AAA_\infty) \ar[r]^-{\eta_{Q(\AAA_\infty)}} & w^\ast w_! Q( \AAA_\infty)
\ar[rr]^-{\text{fibrant replacement}} && w^\ast R w_! Q (\AAA_\infty)
}
\end{flalign}
is a weak equivalence, where $\eta$ is the unit of the adjunction and we also 
performed a fibrant replacement $R$ in the second step. 
(Because every object in $\QFT(\ovr{\CC})$ is fibrant, one could also
drop the fibrant replacement here.) In particular, 
$Q (\AAA_\infty)$, and hence also our original theory $\AAA_\infty \in \QFT_w(\ovr{\CC})$,
is weakly equivalent to the image  under $w^\ast$ of 
the {\em strict} quantum field theory $R w_! Q (\AAA_\infty) \in \QFT(\ovr{\CC})$
(or $w_! Q (\AAA_\infty) \in \QFT(\ovr{\CC})$ if we drop the fibrant replacement).
This defines a model for the strictification of our original homotopy 
quantum field theory $\AAA_\infty \in  \QFT_w(\ovr{\CC})$. 
\end{rem}

\begin{rem}
We would like to issue a warning that Theorem \ref{theo:strictification}
and Remark \ref{rem:srtictification} should {\em not} be misunderstood
as the statement that homotopy quantum field theories are not useful.
Even though it is true that each homotopy quantum field
theory admits a strictification, it is highly non-trivial
to compute such strictifications in practice. The reason
is that in particular the cofibrant replacement in the construction 
of Remark \ref{rem:srtictification} is typically extremely complicated to compute explicitly.
Furthermore, as we will show in Sections \ref{subsec:stacks} and \ref{sec:orbifoldization} below,
interesting constructions naturally lead to non-strict homotopy quantum
field theories, which justifies their practical relevance.
\end{rem}

\subsection{\label{subsec:Einftyresolution}The $E_\infty$-resolution}
We develop a particular functorial $\Sigma$-cofibrant resolution 
for our colored operads $\O_{\ovr{\CC}}$ that is obtained by a component-wise
tensor product with the (chain version of the) Barratt-Eccles $E_\infty$-operad.
This choice of resolution is motivated by the fact that it is not only 
relatively simple, but also sufficiently flexible to encompass 
the examples from Section \ref{subsec:stacks} and \ref{sec:orbifoldization}
in terms of its algebras.
\sk

Let us denote by $\E_\infty\in\Op_{1}(\Ch(k))$ the 
(chain version of the) $1$-colored Barratt-Eccles $E_\infty$-operad 
studied in \cite{BergerFresse}. This operad is defined by
applying the normalized chain complex functor 
to the original simplicial Barratt-Eccles operad. In our work we 
do not need an explicit description of this operad and refer to 
\cite[Section 1.1]{BergerFresse} for the details. Let us recall
some basic properties of the operad  $\E_\infty\in\Op_{1}(\Ch(k))$
that will be important in what follows: 1.)~The operad $\E_\infty$ provides
a $\Sigma$-cofibrant resolution
\begin{flalign}\label{eqn:Einfty2Com}
w : \E_\infty \longrightarrow \mathsf{Com}
\end{flalign}
of the commutative operad $\mathsf{Com}\in \Op_{1}(\Ch(k))$.
2.)~The degree $0$ part of $\E_\infty$ is the associative operad
$\mathsf{As} \in\Op_{1}(\Ch(k))$, i.e.\ there exists an $\Op_{1}(\Ch(k))$-morphism
\begin{flalign}\label{eqn:Ass2Einfty}
i : \mathsf{As} \longrightarrow \E_\infty\quad.
\end{flalign}
3.)~The composition of \eqref{eqn:Ass2Einfty} and \eqref{eqn:Einfty2Com}
is the canonical $\Op_{1}(\Ch(k))$-morphism $\mathsf{As}\to\mathsf{Com}$.
\begin{rem}
Algebras $A\in \Alg(\E_\infty)$ over the operad $\E_\infty\in\Op_{1}(\Ch(k))$
are differential graded algebras, together with additional higher 
chain homotopy data (living in positive degrees in the chain complexes 
underlying $\E_\infty$) that describe homotopy-coherent commutativity.
The two operad morphisms in \eqref{eqn:Einfty2Com} and \eqref{eqn:Ass2Einfty}
yield a sequence of Quillen adjunctions
\begin{flalign}
\xymatrix{
\Alg(\mathsf{As}) ~\ar@<0.5ex>[r]^-{i_!} & \ar@<0.5ex>[l]^-{i^\ast}  ~\Alg(\E_\infty) ~ \ar@<0.5ex>[r]^-{w_!}& \ar@<0.5ex>[l]^-{w^\ast} ~ \Alg(\mathsf{Com})
}\quad.
\end{flalign}
The right adjoint $i^\ast$ assigns to an $\E_\infty$-algebra $A\in \Alg(\E_\infty)$
its underlying differential graded algebra, i.e.\ it forgets the higher chain homotopy data
describing homotopy-coherent commutativity. (Hence, the underlying
differential graded algebra is in general non-commutative.) The right 
adjoint $w^\ast$ assigns to a commutative differential graded algebra 
$A\in \Alg(\mathsf{Com})$ the $E_\infty$-algebra whose chain 
homotopy data are trivial, i.e.\ a strictly commutative $E_\infty$-algebra.
\end{rem}

For every orthogonal category $\ovr{\CC}$, let us define the colored operad
\begin{flalign}
\O_{\ovr{\CC}}\otimes \E_\infty \in\Op_{\CC_0}(\Ch(k))
\end{flalign}
by a component-wise tensor product of the quantum field theory operad from Definition
\ref{def:AQFToperads} and the $1$-colored operad $\E_\infty\in\Op_{1}(\Ch(k))$.
More concretely, the chain complex of operations for
$(\und{c},t) = ((c_1,\dots,c_n),t)\in\CC_0^{n+1}$ reads as
\begin{flalign}\label{eqn:componentsOCCresolved}
\big(\O_{\ovr{\CC}}\otimes \E_\infty \big)\big(\substack{t \\ \und{c}}\big) \,:=\, 
\O_{\ovr{\CC}}\big(\substack{t \\ \und{c}}\big) \otimes \E_\infty(n)\,\cong\,
 \big(\Sigma_n \times \CC(\und{c},t)\big)\big/{\sim_\perp}\otimes \E_\infty(n)\,\in\,\Ch(k)
\end{flalign}
and the operad structure is the tensor product of the respective operad structures.
The operad morphism in \eqref{eqn:Einfty2Com} defines
an $\Op_{\CC_0}(\Ch(k))$-morphism
\begin{flalign}\label{eqn:tmpresolution}
\xymatrix@C=4em{
\O_{\ovr{\CC}}\otimes \E_\infty \ar[r]^-{\O_{\ovr{\CC}}\otimes w} ~&~ \O_{\ovr{\CC}}\otimes \mathsf{Com} \,\cong\, \O_{\ovr{\CC}}\quad,
}
\end{flalign}
where the last isomorphism is due to $\mathsf{Com}(n) = k\in\Ch(k)$, for all $n\geq 0$.
In order to simplify our notations, in the following we will denote \eqref{eqn:tmpresolution} by
\begin{flalign}\label{eqn:wCCresolution}
w_{\ovr{\CC}} : \O_{\ovr{\CC}}\otimes \E_\infty\longrightarrow \O_{\ovr{\CC}}\quad.
\end{flalign}
\begin{theo}\label{theo:Einftyresolution}
Let us assume as before that $k\supseteq \bbQ$. 
For every orthogonal category $\ovr{\CC}$, the $\Op_{\CC_0}(\Ch(k))$-morphism
\eqref{eqn:wCCresolution} defines a $\Sigma$-cofibrant resolution
of the quantum field theory operad $\O_{\ovr{\CC}}\in \Op_{\CC_0}(\Ch(k))$. These resolutions are
functorial in the sense that $w : \O_{(-)}\otimes\E_\infty \to \O_{(-)}$ is a natural transformation
between the functors $\O_{(-)} \otimes \E_\infty :\OCat \to \Op(\Ch(k)) $
and $\O_{(-)}: \OCat \to \Op(\Ch(k))$.
\end{theo}
\begin{proof}
Using the explicit definition in \eqref{eqn:componentsOCCresolved}, we observe that each component
$\big(\O_{\ovr{\CC}}\otimes \E_\infty \big)\big(\substack{t \\ \und{c}}\big)$  is a cofibrant
chain complex because $\E_\infty(n)$ is cofibrant and the $\Set$-tensoring 
(i.e.\ forming coproducts \eqref{eqn:tensoringgeneral}) preserves cofibrant objects.
It follows from Proposition \ref{propo:componentcofimpliesSigmacof}
that $\O_{\ovr{\CC}}\otimes \E_\infty  \in \Op_{\CC_0}(\Ch(k))$ is $\Sigma$-cofibrant.
\sk

We next have to prove that \eqref{eqn:wCCresolution} is an acyclic fibration
in $\Op_{\CC_0}(\Ch(k))$, i.e.\ each component
$w_{\ovr{\CC}} = \id\otimes w : \O_{\ovr{\CC}}\big(\substack{t \\ \und{c}}\big)\otimes \E_\infty(n)
\to \O_{\ovr{\CC}}\big(\substack{t \\ \und{c}}\big)\otimes k$ is an acyclic fibration
in $\Ch(k)$, cf.\ Theorem \ref{theo:Opmodelcat} and Theorem \ref{theo:Chmodelcat}.
The fibration property is clear because each $w : \E_\infty(n)\to k$ is degree-wise surjective.
Using further that each $w : \E_\infty(n)\to k$ is a weak equivalence between
cofibrant objects in $\Ch(k)$, the left Quillen functor property of the tensor product
$V\otimes (-): \Ch(k) \to \Ch(k)$ (for $V\in\Ch(k)$ cofibrant) implies via Ken Brown's lemma \cite[Lemma 1.1.12]{Hovey} 
that $w_{\ovr{\CC}} = \id\otimes w$ is a weak equivalence too.
\sk

Functoriality of these resolutions is obvious.
\end{proof}

\begin{rem}\label{rem:hAQFTlimitingcases}
In order to obtain some intuition for the algebras over $\O_{\ovr{\CC}}\otimes \E_\infty$, i.e.\ 
homotopy quantum field theories $\QFT_w(\ovr{\CC}) := \Alg(\O_{\ovr{\CC}}\otimes \E_\infty )$,
we consider some special limiting cases.
Given any orthogonal category $\ovr{\CC} = (\CC,\perp)$, let us consider as an auxiliary concept
the orthogonal category $\ovr{\CC}^{\max} = (\CC,\perp^{\max})$, where
$\perp^{\max} := \mathrm{Mor}\,\CC \, {}_\mathrm{t}^{}\!\times\!{}^{}_\mathrm{t} \,\mathrm{Mor}\,\CC$
is the maximal orthogonality relation. We have isomorphisms of categories
\begin{flalign}
\QFT(\ovr{\CC}^{\max}) ~\cong~ \Alg(\O_{\ovr{\CC}^{\max}}) ~\cong~\Alg(\mathsf{Com})^{\CC}\quad, 
\end{flalign}
i.e.\ quantum field theories on $\ovr{\CC}^{\max} $ are functors on $\CC$ with values in {\em commutative}
differential graded algebras. For the resolution $\O_{\ovr{\CC}^{\max}} \otimes\E_\infty$, we similarly obtain
isomorphisms of categories
\begin{flalign} \label{aqftmaxortheqn}
\QFT_w(\ovr{\CC}^{\max}) ~=~ \Alg(\O_{\ovr{\CC}^{\max}}\otimes\E_\infty) ~\cong~\Alg(\E_\infty)^{\CC}\quad, 
\end{flalign}
i.e.\ homotopy quantum field theories on $\ovr{\CC}^{\max}$ are functors on $\CC$ with values in 
$\E_\infty$-algebras.
\sk

The orthogonal functor $p:= \id_\CC : \ovr{\CC} \to \ovr{\CC}^{\max}$ induced
by the identity $\id_\CC$ defines a colored operad morphism 
$\O_p \otimes\E_\infty : \O_{\ovr{\CC}}\otimes\E_\infty \to \O_{\ovr{\CC}^{\max}}\otimes\E_\infty$ 
and thus a Quillen adjunction
\begin{flalign}\label{eqn:EinftyQFTs}
\xymatrix{
p_! \,:\, \QFT_w(\ovr{\CC}) ~\ar@<0.5ex>[r]&\ar@<0.5ex>[l]  ~\Alg(\E_\infty)^{\CC}\,:\, p^\ast
}\quad.
\end{flalign}
The right adjoint $p^\ast$ assigns to every functor $\BBB: \CC \to \Alg(\E_\infty)$ with values in 
$\E_\infty$-algebras a homotopy quantum field theory $p^\ast(\BBB)\in \QFT_w(\ovr{\CC})$ 
on $\ovr{\CC}$. Because the latter theories assign homotopy-coherently commutative observable algebras 
to spacetimes, one should interpret them as drastically simplified toy-models
that ignore all quantum theoretic and also Poisson algebraic aspects. We shall show in 
Section \ref{subsec:stacks} below that such theories naturally arise by taking 
suitable ``function algebras'' on $\infty$-stacks, which can be understood as the 
starting point for Poisson algebraic studies of classical gauge theories 
and their deformation quantization.
\sk

Now let us consider as another auxiliary concept the orthogonal category
$\ovr{\CC}^{\min} = (\CC,\emptyset)$ with the trivial orthogonality relation.
We have isomorphisms of categories
\begin{flalign}
\QFT(\ovr{\CC}^{\min}) ~\cong~ \Alg(\O_{\ovr{\CC}^{\min}}) ~\cong~\Alg(\mathsf{As})^{\CC} = \dgAlg(k)^\CC\quad, 
\end{flalign}
i.e.\ quantum field theories on $\ovr{\CC}^{\min} $ are functors on $\CC$ with values in (non-commutative)
differential graded algebras. We define a colored operad morphism
$l : \O_{\ovr{\CC}^{\min}} \to \O_{\ovr{\CC}} \otimes\E_\infty$ by the components 
(recall \eqref{eqn:AQFToperadoperations} and use that the equivalence relation is trivial for $\ovr{\CC}^{\min}$)
\begin{flalign}
\xymatrix{
\ar@{=}[d] \O_{\ovr{\CC}^{\min}}\big(\substack{t \\ \und{c}}\big) \ar[rr]^-{l}~&&~  \O_{\ovr{\CC}}\big(\substack{t \\ \und{c}}\big) \otimes\E_\infty(n) \\
 \ar[d]_-{(\mathrm{diag}\times \id)\otimes \id} \big(\Sigma_n \times \CC(\und{c},t)\big)\otimes k ~&&~\big(\Sigma_n \times \CC(\und{c},t)\big)\big/{\sim_\perp} \otimes \E_\infty(n) \ar[u]_-{\cong}\\
\big(\Sigma_n\times \Sigma_n \times \CC(\und{c},t)\big)\otimes k ~\ar[rr]_-{\cong} &&\big(\Sigma_n \times \CC(\und{c},t)\big)\otimes \mathsf{As}(n) \ar[u]_-{\pi \otimes i}
}
\end{flalign}
where $\mathrm{diag}$ denotes the diagonal map, $\pi$ is the projection to equivalence classes
and $i$ is the operad morphism from \eqref{eqn:Ass2Einfty}. This defines a Quillen adjunction
\begin{flalign}\label{eqn:underlyingdgAlgQFTs}
\xymatrix{
l_! \,:\, \dgAlg(k)^\CC ~\ar@<0.5ex>[r]&\ar@<0.5ex>[l]  ~\QFT_w(\ovr{\CC})\,:\, l^\ast
}\quad.
\end{flalign}
The right adjoint $l^\ast$ assigns to every homotopy quantum field theory $\AAA\in\QFT_w(\ovr{\CC})$
its underlying functor with values in differential graded algebras. Because this constructions
neglects the homotopies encoded in $\E_\infty$, such functor in general does not satisfy 
the strict $\perp$-commutativity axiom. 
\sk

Summing up, we have seen that every homotopy quantum field theory
$\AAA\in \QFT_w(\ovr{\CC})$ has an underlying functor $l^\ast(\AAA) : \CC\to \dgAlg(k)$ that is
obtained by forgetting the chain homotopy data in $\E_\infty$, cf.\ \eqref{eqn:underlyingdgAlgQFTs}.
Furthermore, every functor $\BBB : \CC\to \Alg(\E_\infty)$ with values in $\E_\infty$-algebras
defines a homotopy quantum field theory $p^\ast(\BBB)\in \QFT_w(\ovr{\CC})$, cf.\ \eqref{eqn:EinftyQFTs}. 
Examples of the latter theories are obtained in Corollary \ref{cor:stackQFTs} below. 
\end{rem}

\subsection{\label{subsec:stacks}Toy-models via cochain algebras on stacks}
We provide simple toy-models of homotopy quantum field theories 
on our resolved operads from Theorem \ref{theo:Einftyresolution}
by taking suitable cochain algebras on stacks. (In this section 
the term ``stack'' always refers to ``$\infty$-stacks''.) Applying this
construction to diagrams of stacks will define functors $\CC\to \Alg(\E_\infty)$
with values in $\E_\infty$-algebras, which may be regarded 
as drastically simplified toy-models of homotopy quantum field theories
according to Remark \ref{rem:hAQFTlimitingcases}, see in particular \eqref{eqn:EinftyQFTs}. 
In our opinion these toy-models are
interesting because they provide a convenient starting point 
for investigating Poisson algebraic aspects of classical gauge theories
and their deformation quantization. We however would like
to emphasize that such constructions are technically very involved, see
e.g.\ \cite{DAG,DAG2} for a modern homotopy theoretical approach, 
even in the case of linear fields \cite{BVquantization}.
We therefore postpone these problems to future works
and only outline below the construction of the stacky cochain algebras,
which should be interpreted physically as  ``function algebras'' on stacks.
\sk

In this section we assume that the reader has some familiarity 
with simplicial sets, simplicial $k$-modules and the Dold-Kan correspondence,
see e.g.\ \cite{Schwede} and \cite{GoerssJardine}. 
Before we can address the more complicated case of stacks, let us discuss 
an analog of our construction for simplicial sets, which is 
given by the usual cochain algebra construction.
\sk

We denote by $\sSet$ the category of simplicial sets equipped with the Quillen model structure,
by $\mathbf{sMod}_k$ the model category of simplicial $k$-modules 
(weak equivalences and fibrations are defined as for the underlying simplicial sets)
and by $\Ch_{\geq 0}(k)$ the category of non-negatively graded chain complexes 
with the projective model structure (weak equivalences are quasi-isomorphisms 
and fibrations are surjective in all positive degrees). 
There exists a sequence of Quillen adjunctions (left adjoints point from left to right)
\begin{flalign}\label{eqn:normalizedchains}
\xymatrix@C=4em{
\sSet \ar@<0.5ex>[r]^-{(-)\otimes k}~&~  \ar@<0.5ex>[l]^-{U} \mathbf{sMod}_k\ar@<0.5ex>[r]^-{N_\ast} ~&~ \ar@<0.5ex>[l]^-{\Gamma}  \Ch_{\geq 0}(k) ~ \ar@<0.5ex>[r]^-{\iota}&~ \ar@<0.5ex>[l]^-{\tau_{\geq 0}} \Ch(k)
}\quad.
\end{flalign}
The first step is the free-forget adjunction between simplicial sets and simplicial $k$-modules,
the second step is the Dold-Kan correspondence and the last step is the inclusion-(good)truncation adjunction
for chain complexes. The composite of left adjoints is the functor
that assigns to a simplicial set $X\in \sSet$ its normalized chains $N_\ast(X)\in\Ch(k)$,
where here and in the following we suppress both the free simplicial $k$-module functor
$(-)\otimes k : \sSet\to\mathbf{sMod}_k$ and the inclusion functor $\iota : \Ch_{\geq 0}(k)\to \Ch(k)$.
\sk

The normalized cochains on a simplicial set are obtained by composing the 
normalized chains functor with the internal hom functor $[-,k]$, which is the 
left adjoint in the Quillen adjunction
\begin{flalign}\label{eqn:cotensoringadjunction}
\xymatrix{
[-,k] \,:\, \Ch(k) ~\ar@<0.5ex>[r]&\ar@<0.5ex>[l]  ~\Ch(k)^\op\,:\, k^{(-)}
}\quad.
\end{flalign}
Here the right adjoint functor is given by cotensoring, which in the present case coincides 
with the internal hom  $k^V = [V,k]$, for all $V\in\Ch(k)$. The composition of
the left adjoints in \eqref{eqn:normalizedchains} and \eqref{eqn:cotensoringadjunction} 
defines a left Quillen functor
\begin{flalign}
N^\ast \,:\, \sSet \longrightarrow \Ch(k)^{\op}\quad,
\end{flalign}
which is the normalized cochains functor. Notice that
we do not have to derive this functor  because every
simplicial set is cofibrant. It was shown in \cite{BergerFresse}
that the normalized cochains
on every simplicial set $X\in\sSet$ carry a canonical $\E_\infty$-algebra
structure, i.e.\ 
\begin{flalign}\label{eqn:cochainssSetformula}
N^\ast(X) = [N_\ast(X),k] \in\Alg(\E_\infty)\quad.
\end{flalign}
\begin{ex}\label{ex:sSets}
We may regard any set $S\in\Set$ as a constant simplicial set $S\in\sSet$.
Forming the normalized cochains on this simplicial set, we obtain a
chain complex concentrated in degree $0$ with trivial differential.
Concretely, it is given by $N^\ast(S) = \mathrm{Map}(S,k)$,
i.e.\ the normalized cochain algebra on a set $S\in\Set$ is precisely its
function algebra.
\sk

More interestingly, we may regard any groupoid $\mathcal{G} \in \Grpd$
as a simplicial set via the nerve functor $B(\mathcal{G}) \in\sSet$.
In this case the normalized cochain algebra $N^\ast(B(\mathcal{G}))$
is precisely the groupoid cohomology algebra, see e.g.\ \cite{Crainic}.
\end{ex}

The construction above generalizes to the case of stacks, which
loosely speaking are smooth simplicial sets. Recall from 
e.g.\ \cite{Schreiber,Hollander,Dugger} that stacks may be described 
by presheaves  of simplicial sets on the site of Cartesian spaces, 
i.e.\ they are objects in the category
\begin{flalign}
\mathbf{H}\, := \, \PSh(\Cart,\sSet)\, =\, \sSet^{\Cart^\op}\quad.
\end{flalign}
For an example, we refer to the stack of Yang-Mills fields constructed in \cite{BSSStack}.
\sk

In the following we shall endow  $\mathbf{H}$ (as well as every other category of presheaves with 
values in a model category) with its global projective model structure. We explain in
Remark \ref{rem:localmodelstructure} below how our results can be extended to
stacks, which are the fibrant objects in the local model structure on $\mathbf{H}$, cf.\ \cite{Hollander,Dugger}.
From \eqref{eqn:normalizedchains}, we obtain the following 
induced sequence of Quillen adjunctions between presheaf model categories
\begin{flalign}\label{eqn:normalizedchainsstacks}
\xymatrix@C=2em{
\mathbf{H} \ar@<0.5ex>[r]^-{(-)\otimes k}~&~  \ar@<0.5ex>[l]^-{U} \PSh(\Cart,\mathbf{sMod}_k)
\ar@<0.5ex>[r]^-{N_\ast} ~&~ \ar@<0.5ex>[l]^-{\Gamma}  \PSh(\Cart,\Ch_{\geq 0}(k)) ~ 
\ar@<0.5ex>[r]^-{\iota}&~ \ar@<0.5ex>[l]^-{\tau_{\geq 0}} \PSh(\Cart,\Ch(k))
}\quad,
\end{flalign}
where each functor acts object-wise on presheaves.
Suppressing as before the functors $(-)\otimes k$ and $\iota$, 
the composition of left adjoints in \eqref{eqn:normalizedchainsstacks}
assigns to an object $X\in \mathbf{H} $ the presheaf of chain complexes
$N_\ast(X)\in  \PSh(\Cart,\Ch(k))$ whose value on a test space $T\in\Cart$ is 
\begin{flalign}\label{eqn:Nastcomponents}
N_\ast(X)(T) ~=~N_\ast\big(X(T)\big)\in\Ch(k)\quad.
\end{flalign}

We shall now generalize the Quillen adjunction in \eqref{eqn:cotensoringadjunction} to the case of presheaves.
Because we are working with presheaves on Cartesian spaces, we choose $k =\bbR$ or $k=\bbC$
as the underlying base ring. (With this choice we may describe real or complex valued cochain algebras.)
We define the object $\und{k}\in\PSh(\Cart,\Ch(k))$ by setting
$\und{k}(T) := C^\infty(T,k)$ (concentrated in degree $0$ with trivial differential), for
all test spaces $T \in\Cart$. Because the category 
$\PSh(\Cart,\Ch(k))$ is enriched over $\Ch(k)$, we have a mapping chain complex functor
$[-,-]^\infty : \PSh(\Cart,\Ch(k))^\op\times \PSh(\Cart,\Ch(k))\to \Ch(k)$, which is explicitly
given by the end formula
\begin{flalign}\label{eqn:inftyhom}
[V,W]^\infty = \int_{T\in\Cart^\op} \big[V(T),W(T)\big] \quad,
\end{flalign}
for all $V,W\in \PSh(\Cart,\Ch(k))$, where on the right-hand side $[-,-]$ is the internal hom in $\Ch(k)$.
There exists an adjunction
\begin{flalign}\label{eqn:cotensoringadjunctionstacks}
\xymatrix{
[-,\und{k}]^\infty \,:\, \PSh(\Cart,\Ch(k)) ~\ar@<0.5ex>[r]&\ar@<0.5ex>[l]  ~\Ch(k)^\op\,:\, \und{k}^{(-)}
}\quad,
\end{flalign}
where the right adjoint is given by cotensoring
$\und{k}^V (T) = [V,\und{k}(T)]$, for all $V\in\Ch(k)$ and $T\in\Cart$.
Because $\und{k}\in \PSh(\Cart,\Ch(k))$ is a fibrant object, it follows
that \eqref{eqn:cotensoringadjunctionstacks} is a Quillen adjunction.
The composition of the left adjoints in \eqref{eqn:normalizedchainsstacks}
and \eqref{eqn:cotensoringadjunctionstacks} defines a left Quillen functor
\begin{flalign}
{N^{\infty}}^{\ast} \, : \, \mathbf{H} \longrightarrow \Ch(k)^\op \quad,
\end{flalign}
which we call the {\em smooth normalized cochains functor}.
\sk

In contrast to the case of simplicial sets above, the smooth normalized cochains functor
does not necessarily preserve weak equivalences and hence it has to be derived.
In order to provide a concrete construction, we can take the very explicit 
cofibrant replacement functor $Q : \mathbf{H}\to\mathbf{H}$ developed by 
Dugger in \cite{DuggerUniversal} and define a derived functor by pre-composition
\begin{flalign}\label{eqn:Nastderived}
\bbL {N^{\infty}}^{\ast} \,:= \,  {N^{\infty}}^{\ast} \, Q\, : \, \mathbf{H} \longrightarrow \Ch(k)^\op \quad.
\end{flalign}
For every $X\in \mathbf{H}$, the derived smooth normalized 
cochains on $X$ carry a canonical $\E_\infty$-algebra structure, i.e.\ we canonically have that
\begin{flalign}
\bbL  {N^{\infty}}^{\ast}(X) \, = \, \big[N_\ast\big(Q(X) \big) ,\und{k}\big]^\infty \in\Alg(\E_\infty)\quad.
\end{flalign}
The relevant argument goes as follows:
Recalling \eqref{eqn:inftyhom} and \eqref{eqn:Nastcomponents},
we have that
\begin{flalign}
\bbL {N^{\infty}}^{\ast} (X) = \int_{T\in\Cart^\op} \big[ N_\ast\big(Q(X)(T) \big),\und{k}(T)\big] \quad.
\end{flalign}
Using that by \cite{BergerFresse} each $N_\ast\big(Q(X)(T)\big)\in\Ch(k)$ carries a canonical 
$\E_\infty$-{\em co}algebra structure, we obtain that each term under the end is canonically
an $\E_\infty$-algebra and hence so is the end.
\begin{rem}\label{rem:localmodelstructure}
Recall that in our constructions above we have endowed the category
$\mathbf{H} = \PSh(\Cart,\sSet)$ with its global projective model structure.
However, for the study of stacks the local projective model structure on $\mathbf{H}$
is more suitable, see \cite{Hollander,Dugger}. In short, the local model structure
is obtained via left Bousfield localization of the global model structure at all hypercovers.
This implies that local and global cofibrations are precisely the same, but there are more
weak equivalences and fewer fibrations in the local model structure. We denote
the local projective model structure by $\mathbf{H}_{\mathrm{loc}}$
and recall that the category of stacks is by definition the full subcategory 
$ \mathbf{St}\subseteq \mathbf{H}_{\mathrm{loc}}$ of locally fibrant objects. 
Notice that the full subcategory  $\mathbf{St}\subseteq \mathbf{H}_{\mathrm{loc}}$
is not necessarily a model category in its own right. However it is a so-called
{\em homotopical category} in the sense of \cite{Riehl}, i.e.\
a category with a notion of weak equivalences satisfying the $2$-out-of-$6$ property. 
As a consequence of  \cite[Proposition 3.3.5]{Hirschhorn}, we obtain that the local weak equivalences
between stacks are precisely the global weak equivalences. This implies that
\eqref{eqn:Nastderived} restricts to a homotopical functor
\begin{flalign}
\bbL {N^{\infty}}^{\ast}  \,:\, \mathbf{St} \longrightarrow \Ch(k)^\op
\end{flalign}
on the homotopical category of stacks, i.e.\ it preserves local weak equivalences between stacks.
\end{rem}

We conclude this subsection with an application to homotopy quantum field theory.
\begin{cor}\label{cor:stackQFTs}
Let $\ovr{\CC}$ be an orthogonal category. For every functor
$X : \CC^\op\to \mathbf{St}\subseteq \mathbf{H}_{\mathrm{loc}}$ 
with values in stacks (i.e.\ fibrant objects), we obtain a functor 
$\bbL  {N^{\infty}}^{\ast}(X) : \CC\to \Alg(\E_\infty)$, 
i.e.\ via \eqref{eqn:EinftyQFTs} an example of a homotopy quantum field
theory $p^\ast \big(\bbL {N^{\infty}}^{\ast} (X)\big)\in \QFT_w(\ovr{\CC})$.
\end{cor}

\begin{ex}
Let us consider the orthogonal category $\ovr{\Loc}$ from Example \ref{ex:LocOCat}.
It was shown in  \cite{BSSStack}  that for each spacetime $M\in\Loc$
there exists a stack $\mathsf{YM}_G(M) \in \mathbf{St}$ of solutions of the 
Yang-Mills equation with a (possibly non-Abelian) structure group $G$.
This assignment is contravariantly functorial, i.e.\ we have a functor 
$\mathsf{YM}_G : \Loc^\op \to\mathbf{St}$. We interpret the
derived smooth normalized cochain algebra $\bbL {N^{\infty}}^{\ast}(\mathsf{YM}_G(M) ) \in \Alg(\E_\infty)$
as a higher algebra of observables for classical Yang-Mills theory on the spacetime $M$. The corresponding functor
$\bbL  {N^{\infty}}^{\ast}(\mathsf{YM}_G) : \Loc\to\Alg(\E_\infty) $ then defines via  \eqref{eqn:EinftyQFTs}
a toy-model for a homotopy quantum field theory. (Let us recall from above that this construction
does not yet take into account Poisson geometric aspects and/or deformation quantization of the gauge theory.)
\sk

It is important to emphasize that our higher algebras of observables
$\bbL  {N^{\infty}}^{\ast}(\mathsf{YM}_G(M) ) \in  \Alg(\E_\infty)$ are valued in chain complexes.
They describe more than just gauge invariant observables, which are given by
the $0$th homology. Our algebras 
should rather be interpreted as a generalization of the Chevalley-Eilenberg algebras 
(i.e.\ the BRST formalism in physics) from infinitesimal to finite gauge transformations. 
For (finite-dimensional) Lie groupoid cochain algebras, a precise relationship is 
understood via the van Est map, see e.g.\ \cite{Crainic}. We expect that a similar
result holds true for stacks and will come back to this issue in a future work.
\end{ex}


\section{\label{sec:orbifoldization}Examples via homotopy invariants}
In this section we present another class of examples 
of non-strict homotopy quantum field theories on our resolved operads
from Theorem \ref{theo:Einftyresolution}.  Our envisaged 
construction already appeared in a less formal and complete 
approach in a previous work of two of us  \cite{BeniniSchenkel} and it can be 
interpreted in terms of orbifoldization, i.e.\ it takes homotopy invariants 
of local groupoid actions on quantum field theories. See also
\cite{SW} for a similar construction for topological field theories. The main missing point
in \cite{BeniniSchenkel} was that we could only establish the relevant homotopy 
data in low orders, but we had no control of its coherence. It is due to the new
techniques developed in the present paper that we can now 
prove homotopy-coherence of the construction proposed in \cite{BeniniSchenkel}.
In Example \ref{ex:fiberedcat} below we explain that
such constructions might also be relevant for perturbative quantum gauge theories.
Let us emphasize that, in contrast to our previous examples obtained from cochain 
algebras on stacks (cf.\ Corollary \ref{cor:stackQFTs}), 
the homotopy quantum field theories
obtained from orbifoldization in general {\em do not}
admit a description in terms of an $\E_\infty$-algebra valued
functor, hence they encode quantum theoretic features.
\sk

Let $\ovr{\CC}$ be an orthogonal category and $\pi : \DD\to \CC$ 
a category fibered in groupoids. We endow $\DD$ with the pullback
orthogonality relation $\perp_\DD := \pi^\ast(\perp_\CC)$,
i.e.\ $g_1\perp_\DD g_2$ if and only if $\pi(g_1)\perp_\CC\pi(g_2)$.
We obtain an orthogonal functor $\pi : \ovr{\DD}\to \ovr{\CC} $,
which we call an {\em orthogonal category fibered in groupoids}.
Given any strict $\Ch(k)$-valued quantum field theory
$\AAA \in \QFT(\ovr{\DD})$ on the total category $\ovr{\DD}$, we would like to
define a quantum field theory on the base category $\ovr{\CC}$
by forming {\em homotopy invariants} along the groupoid fibers
$\pi^{-1}(c)\in\Grpd$, for all $c\in\CC$. Since forming (homotopy) invariants
corresponds to categorical (homotopy) limits, we are looking for something like
a (derived) right adjoint functor of the pullback functor $\pi^\ast : \QFT(\ovr{\CC})\to \QFT(\ovr{\DD})$.
Unfortunately, such right adjoints (i.e.\ operadic {\em right} Kan extensions)
generically do not exist, in contrast to the left adjoints (i.e.\ operadic {\em left} Kan 
extensions) in \eqref{eqn:generaladjunction}. We therefore propose the following
alternative construction: Consider the underlying chain complex valued functor
$\AAA : \DD\to \Ch(k)$ of our quantum field theory $\AAA\in \QFT(\ovr{\DD})$.
In contrast to the situation above, the pullback functor
$\pi^\ast :\Ch(k)^\CC \to \Ch(k)^\DD$ on functor categories admits a right adjoint,
which is given by the right Kan extension $\Ran_\pi : \Ch(k)^\DD\to\Ch(k)^\CC$
along $\pi:\DD\to \CC$. The derived functor corresponding to the right Kan extension
is called {\em homotopy right Kan extension} (cf.\  \cite{hoKan2, hoKan3})
and it will be denoted as usual by $\hoRan_\pi : \Ch(k)^\DD\to\Ch(k)^\CC$.
Forming the homotopy right Kan extension of (the underlying $\Ch(k)$-valued functor of) 
our quantum field theory $\AAA\in \QFT(\ovr{\DD})$ however only defines
an object $\hoRan_\pi \AAA \in \Ch(k)^{\CC}$, i.e.\ we do not automatically
get the structure of a quantum field theory on $\ovr{\CC}$.  As observed in
\cite{BeniniSchenkel} by direct computations, it is generically not true
that  $\hoRan_\pi \AAA \in \Ch(k)^{\CC}$ carries the structure
of a {\em strict} quantum field theory on $\ovr{\CC}$ because
especially the $\perp$-commutativity property only holds true up to
chain homotopies. The novel result which we shall prove in this section
is as follows: Choosing a strictified model (in the sense of \cite{Hollander}) for the 
orthogonal category fibered in groupoids $\pi : \ovr{\DD}\to \ovr{\CC} $,
the homotopy right Kan extension $\hoRan_\pi \AAA$ carries canonically
the structure of a homotopy quantum field theory on $\ovr{\CC}$
corresponding to the resolution $w_{\ovr{\CC}} : \O_{\ovr{\CC}}\otimes \E_\infty \to \O_{\ovr{\CC}}$
from Theorem \ref{theo:Einftyresolution}, i.e.\ $\hoRan_\pi \AAA\in\QFT_w(\ovr{\CC})$.
\begin{ex}\label{ex:fiberedcat} 
Let $\ovr{\Loc}$ be the orthogonal category of 
spacetimes from Example \ref{ex:LocOCat}. 
For constructing  perturbative quantum Yang-Mills theory 
(with non-Abelian structure group $G$) within the 
BRST/BV formalism for algebraic quantum field theory,
one has to choose background principal $G$-bundles with
connections (satisfying the classical non-linear Yang-Mills equation)
around which one can perturb. This implies that the natural 
assignment of observable algebras for perturbative quantum Yang-Mills theory 
is {\em not} to spacetimes $M\in\Loc$ but rather
to triples $(M,P,A)$, where $M\in\Loc$ is a spacetime and $(P,A)$ is a Yang-Mills field on
$M$, i.e.\ a principal $G$-bundle $P\to M$ with connection $A$ satisfying the Yang-Mills equation.
\sk

Let us denote the category of such triples by $\mathbf{YM}_G\Loc$. A morphism
$(M,P,A)\to (M^\prime,P^\prime,A^\prime)$ is a principal $G$-bundle
morphism $g : P \to P^\prime$ that induces a $\Loc$-morphism $M\to M^\prime$
on the base spaces and that preserves the connections, i.e.\  $g^\ast A^\prime =A$.
(A similar category appeared before in the work of Zahn on quantum field theories
in the presence of background gauge fields \cite{Zahn}.)
Notice that there exists an obvious projection functor
$\pi : \mathbf{YM}_G\Loc \to\Loc$ which assigns to a triple $(M,P,A)$
its underlying spacetime $M$ and to a morphism $(M,P,A)\to (M^\prime,P^\prime,A^\prime)$ 
its underlying $\Loc$-morphism $M\to M^\prime$. It is easy to check
that $\pi : \mathbf{YM}_G\Loc \to\Loc$  is a category fibered in groupoids
with fibers $\pi^{-1}(M) \in\Grpd$ the groupoids of all
Yang-Mills fields $(P,A)$ over $M\in\Loc$, which we interpret as background fields. 
Pulling back the orthogonality relation from $\ovr{\Loc}$,
we obtain an orthogonal category fibered in groupoids 
$\pi : \ovr{\mathbf{YM}_G\Loc} \to \ovr{\Loc}$.
\sk

The constructions in \cite{Hollands,FredenhagenRejzner,FredenhagenRejzner2} and 
in particular in \cite{TehraniZahn} concretely define perturbative quantum Yang-Mills theory 
in terms of a functor $\AAA : \mathbf{YM}_G\Loc\to \dgAlg(k) $ that satisfies
the relevant $\perp$-commutativity axiom, 
i.e.\ $\AAA \in \QFT(\ovr{ \mathbf{YM}_G\Loc})$.
This is however not completely satisfactory
for the following reason: For a fixed spacetime $M\in\Loc$,
we do not only obtain a single algebra of observables,
but rather a whole diagram $\AAA\vert_{\pi^{-1}(M)}^{} : \pi^{-1}(M)\to \dgAlg(k)$
of observable algebras corresponding to different choices
of background fields around which we perturb. 
It would be desirable to construct from this data a single 
observable algebra $\AAA_\pi (M)\in\dgAlg(k)$ that 
combines both the background gauge fields on $M$
and the quantum perturbations around each background.
For constructing such a single algebra one has to remove certain 
gauge redundancies whose origin lies in the following classical picture: 
Perturbations of Yang-Mills fields on a spacetime $M\in \Loc$ admit different 
presentations in terms of gauge equivalent 
background Yang-Mills fields. These presentations should 
be regarded as being equivalent and hence
identified by taking a suitable (homotopy) 
quotient of the stack of Yang-Mills fields with perturbations.
At the dual level of observable algebras, this amounts to 
forming the homotopy limit 
\begin{flalign}
\AAA_\pi (M) \,:=\, \holim \Big(\AAA\vert_{\pi^{-1}(M)}^{} : \pi^{-1}(M)\to \dgAlg(k)\Big) \in \dgAlg(k)
\end{flalign}
of the diagram of observables on $M$. This means that we are taking homotopy invariants along
the groupoid fibers $\pi^{-1}(M)$. By construction, this defines a single algebra $\AAA_\pi (M)\in\dgAlg(k)$
that combines the observables for the background gauge fields on $M$ and the observables for the quantum
perturbations around each background in a way that the gauge redundancies mentioned above are removed.
The main aim of this section is to generalize and formalize this construction and to show that it defines, 
after a suitable strictification of the category fibered in groupoids, a homotopy 
quantum field theory, i.e.\ an algebra over our resolved operad from Theorem \ref{theo:Einftyresolution}.
An open problem, which is not addressed in this paper, is to compute explicitly the 
homotopy invariants $\AAA_\pi (M)\in\dgAlg(k)$ for perturbative quantum Yang-Mills theory 
and to provide a physical interpretation of the homology groups of this differential graded algebra.
We however would like to refer the reader to \cite{BeniniSchenkel} for concrete computations 
in the simpler setting where the input theory $\AAA$ consists of chain complexes 
concentrated in degree $0$.
\end{ex}

In the following we let $\ovr{\CC}$ be any orthogonal category. Because
of the strictification result in \cite[Theorems 1.2 and 1.4]{Hollander}, 
we may work without loss of generality with presheaves of groupoids
$F:\CC^\op\to \Grpd$ instead of categories fibered in groupoids over $\CC$. 
(Example \ref{ex:strictmodel} below explains how to rephrase our Example 
\ref{ex:fiberedcat} in terms of a presheaf.)  We recall that every presheaf
$F \in \PSh(\CC,\Grpd)$ defines a category fibered in groupoids
via the {\em Grothendieck construction}: The total category
$\CC_F$ is the category whose objects are pairs
$(c,x)$ with $c\in\CC$ and $x\in F(c)$ and whose morphisms
are pairs $(f, h) : (c,x)\to (c^\prime,x^\prime)$ 
with $f : c\to c^\prime$ in $\CC$ and $h : x \to f^\ast x^\prime$ in $F(c)$.
For better readability, here and in the following we use the pullback notation $f^\ast := F(f) 
: F(c^\prime)\to F(c)$ for the functor corresponding to a $\CC$-morphism 
$f : c\to c^\prime$. The composition of two $\CC_F$-morphisms $(f, h) : (c,x)\to (c^\prime,x^\prime)$ 
and $(f^\prime , h^\prime ) : (c^\prime,x^\prime)\to (c^{\prime\prime},x^{\prime\prime})$
is given by $(f^\prime , h^\prime )\, (f, h) := \big(f^\prime\,f, (f^\ast h^\prime)\, h\big) :(c,x) 
\to (c^{\prime\prime},x^{\prime\prime})$. The obvious projection functor
$\pi : \CC_F \to \CC$, given by $(c,x) \mapsto c$ and $(f,h)\mapsto f$,
defines a category fibered in groupoids whose fiber $\pi^{-1}(c)$ over $c\in \CC$ is
naturally isomorphic to the value $F(c)$ of the presheaf 
$F :\CC^\op\to \Grpd$ on $c$. When endowed with the
pullback orthogonality relation, we obtain an orthogonal category 
fibered in groupoids $\pi : \ovr{\CC}_F\to \ovr{\CC}$.
\begin{ex}\label{ex:strictmodel} 
Recall from Example \ref{ex:fiberedcat} 
the category fibered in groupoids $\pi : \mathbf{YM}_G\Loc \to \Loc$
that describes Yang-Mills fields over spacetimes. The particular model
we have given above is not strict because pullbacks
of principal $G$-bundles and connections along principal
bundle morphisms are only pseudo-functorial, i.e.\
the assignment of fibers $\Loc\ni M \mapsto \pi^{-1}(M)\in\Grpd$ 
is just a contravariant pseudo-functor. A weakly equivalent strict model
has been developed in \cite{BSSStack}. More precisely, this
paper constructs a functor $\mathsf{YM}_G : \Loc^\op \to \mathbf{St}$
to the category of stacks such that $\mathsf{YM}_G(M)$ is the moduli stack
of Yang-Mills fields on $M$. Because these Yang-Mills stacks are just $1$-stacks
(in contrast to $\infty$-stacks), we obtain a groupoid of Yang-Mills fields
on every $M\in\Loc$ by taking the global points of these stacks, i.e.\
$F(M) := \mathsf{YM}_G(M)(\ast) \in\Grpd$. This defines
a presheaf of groupoids $F : \Loc^\op\to \Grpd$ whose Grothendieck
construction is weakly equivalent (in the sense of \cite{Hollander}) to
the category fibered in groupoids $\pi : \mathbf{YM}_G\Loc \to \Loc$.
Let us mention that the basic reason why the construction in \cite{BSSStack} defines
a strict presheaf is that the groupoid of Yang-Mills fields $F(M) $ on any 
$M\in\Loc$ is described in terms of {\v C}ech data subordinate to
the canonical cover given by {\em all} diamond subsets $U\subseteq M$ of $M$.
\end{ex}

We shall now present our model for the homotopy
right Kan extension $\hoRan_\pi : \Ch(k)^{\CC_F} \to \Ch(k)^\CC$
along the projection functor $\pi : \CC_F\to \CC$.
\begin{propo}\label{propo:hoRan}
Let $F\in\PSh(\CC,\Grpd)$ be a presheaf of groupoids and $\AAA :\CC_F\to\Ch(k)$
a chain complex valued functor on the corresponding Grothendieck construction.
Then the following formula defines a model for the homotopy right Kan extension
$\hoRan_\pi \AAA: \CC\to\Ch(k)$ of $\AAA$ along the projection functor
$\pi : \CC_F\to\CC$. For all $c\in\CC$,
\begin{flalign}\label{eqn:hoRanend}
\hoRan_\pi\AAA(c) \,:=\, \int_{x\in F(c)} \big[ N_\ast \big(B\big(F(c) \downarrow x\big)\big) , \AAA(c,x)\big] \quad,
\end{flalign}
where $[-,-]: \Ch(k)^\op\times \Ch(k) \to \Ch(k)$ is the internal hom functor,
$B: \Cat\to \sSet$  is the nerve functor, $N_\ast :\sSet \to\Ch(k)$ is the normalized chain complex functor
(i.e.\ the composition of left adjoints in \eqref{eqn:normalizedchains}) and
$F(c) \downarrow x$ is the over category of $F(c)$ over $x\in F(c)$.
\end{propo}
\begin{proof}
Recall from \cite{hoKan2, hoKan3} that the homotopy right Kan extension 
can be computed point-wise as a homotopy limit: For every $c\in\CC$,
the chain complex $\hoRan_\pi \AAA (c)\in\Ch(k)$ is the homotopy limit
of the diagram
\begin{flalign}
\xymatrix{
c \downarrow \pi \ar[r] & \CC_F \ar[r]^-{\AAA}& \Ch(k) 
}\quad,
\end{flalign}
where $c \downarrow \pi$ is the under category 
of $\pi : \CC_F\to \CC$ under $c\in\CC$.
In the proof of \cite[Theorem 5.3]{BeniniSchenkel}
it was shown that the functor $F(c) \to c \downarrow \pi$
defined by $x\mapsto \big(\id_c : c \to \pi(c,x)\big)$ is homotopy
initial, i.e.\ there exists a weakly equivalent but simpler model where
\begin{flalign}\label{eqn:holimtmp}
\hoRan_\pi \AAA (c) \,= \, \holim\Big(\AAA\vert_{\pi^{-1}(c)} : F(c)\to \Ch(k)\Big)
\end{flalign}
is the homotopy limit of the restriction of $\AAA :\CC_F\to\Ch(k)$ 
to the fiber $\pi^{-1}(c)\cong F(c)$.
\sk

We compute \eqref{eqn:holimtmp} by using the standard Bousfield-Kan formula, 
see e.g.\ \cite[Chapter 19]{Hirschhorn} for an excellent summary.
For this we endow the model category $\Ch(k)$ with the framing
determined by the functors $\Ch(k)\ni V\mapsto N_\ast(\Delta[-])\otimes V \in \Ch(k)^\Delta$
and $\Ch(k) \ni V\mapsto [N_\ast(\Delta[-]),V]\in \Ch(k)^{\Delta^\op}$.
Here $\Delta$ is the usual simplex category, hence $ \Ch(k)^\Delta$ is the category of
cosimplicial chain complexes and $\Ch(k)^{\Delta^\op}$ is that of simplicial chain complexes.
Moreover, $\Delta[-] : \Delta \to \sSet$ is the functor that
assigns to $[n]\in\Delta$ the simplicial $n$-simplex $\Delta[n]\in\sSet$.
The Bousfield-Kan formula in \cite[Definition 19.1.5]{Hirschhorn} can be 
written as the end
\begin{flalign}
 \holim\Big(\AAA\vert_{\pi^{-1}(c)} : F(c)\to \Ch(k)\Big) \,=\,  \int_{x\in F(c)} \widehat{\AAA}(c,x)^{B(F(c)\downarrow x)}\quad,
\end{flalign}
where $\widehat{\AAA}(c,x) :=  \big[N_\ast(\Delta[-]),\AAA(c,x)\big] \in \Ch(k)^{\Delta^\op}$ 
is the simplicial frame on $\AAA(c,x)\in\Ch(k)$. The expression $\widehat{\AAA}(c,x)^{B(F(c)\downarrow x)}$ under the
end is defined in \cite[Definition 16.3.1]{Hirschhorn} as the
limit of the diagram
\begin{flalign}
\xymatrix@C=3em{
\big( \Delta[-]\downarrow B\big( F(c)\downarrow x\big)\big)^\op \ar[r] & \Delta^\op \ar[rr]^-{[N_\ast(\Delta[-]),\AAA(c,x)]}&&
\Ch(k)
}\quad.
\end{flalign}
Introducing for notational simplicity the category 
$J := \Delta[-]\downarrow B\big( F(c)\downarrow x\big)$, we compute
\begin{flalign}
\nn \widehat{\AAA}(c,x)^{B(F(c)\downarrow x)} &\,=\, \mathrm{lim}_{J^\op}^{} \Big(\big[N_\ast(\Delta[-]),\AAA(c,x)\big]\Big)
\,\cong\, \big[N_\ast\big(\colim_J^{} \big(\Delta[-]\big)\big),\AAA(c,x)\big]\\
&\,\cong\, \big[ N_\ast \big(B\big(F(c) \downarrow x\big)\big) , \AAA(c,x)\big]\quad.
\end{flalign}
For the first isomorphism we used that both $[-,W] : \Ch(k)\to\Ch(k)^\op$, for any $W\in\Ch(k)$,
and $N_\ast:\sSet\to \Ch(k)$ are left adjoint functors and hence they preserve colimits. (The colimit
of a diagram $J \to \Ch(k)^\op$ is the limit of the opposite diagram $J^\op \to \Ch(k)$.) In the second step
we used that every simplicial set $X\in\sSet$ is a
colimit of simplicial $n$-simplices via 
\begin{flalign}
X\,\cong\, \colim\Big(\xymatrix{
\Delta[-]\downarrow X \ar[r]& \Delta \ar[r]^-{\Delta[-]}& \sSet
}\Big)\quad.
\end{flalign}
This completes our proof.
\end{proof}

\begin{rem}
We would like to emphasize that it is because of \eqref{eqn:holimtmp} 
that the homotopy right Kan extension $\hoRan_\pi \AAA(c)$
describes the homotopy invariants of a functor $\AAA : \CC_F\to \Ch(k)$
along the groupoid fiber $\pi^{-1}(c)\cong F(c)\in\Grpd$. Let us also stress that 
to obtain \eqref{eqn:holimtmp} one uses crucially that $\pi : \CC_F\to\CC$ 
is a category fibered in groupoids, see the proof of 
\cite[Theorem 5.3]{BeniniSchenkel} for the details.
\end{rem}

Now we can state the main theorem of this section.
In the proof below we use the fundamental result by Berger and Fresse 
\cite{BergerFresse} that the normalized chain complex of a simplicial set 
carries a canonical {\em co}action of the Barratt-Eccles operad $\E_\infty\in\Op_{1}(\Ch(k))$.
\begin{theo}\label{theo:hoRanQFT}
Let $\ovr{\CC}$ be an orthogonal category and $F \in \PSh(\CC,\Grpd)$
a presheaf of groupoids. Furthermore, let $\pi : \ovr{\CC}_F \to \ovr{\CC}$ be the
orthogonal category fibered in groupoids that is obtained by the Grothendieck construction applied to $F$.
Given any strict quantum field theory $\AAA\in \QFT(\ovr{\CC}_F)$ on 
the total category $\ovr{\CC}_F$, the family of chain complexes
$\hoRan_\pi\AAA(c)\in\Ch(k)$, for $c\in\CC$, from Proposition \ref{propo:hoRan}
carries canonically the structure of an $\O_{\ovr{\CC}}\otimes\E_\infty$-algebra.
In other words, $\hoRan_\pi\AAA\in \QFT_w(\ovr{\CC})$
is a homotopy quantum field theory on the base category $\ovr{\CC}$
corresponding to our resolution $w_{\ovr{\CC}} : \O_{\ovr{\CC}}\otimes\E_\infty \to \O_{\ovr{\CC}}$ 
from Theorem  \ref{theo:Einftyresolution}. 
\end{theo}
\begin{proof}
We recall from Section \ref{subsec:algebras} that an $\O_{\ovr{\CC}} \otimes \E_\infty$-action
on the family of chain complexes $\hoRan_\pi\AAA(c)\in\Ch(k)$, for $c\in\CC$,
is given by specifying for each $n\geq 0$ and $(\und{c},t)\in\CC_0^{n+1}$
a $\Ch(k)$-morphism
\begin{flalign}\label{eqn:alphacomplete}
\alpha \,:\, \O_{\ovr{\CC}}\big(\substack{t\\ \und{c}}\big) \otimes \E_{\infty}(n) \otimes \hoRan_\pi\AAA(\und{c}) \longrightarrow
\hoRan_\pi\AAA(t)\quad,
\end{flalign}
where $ \hoRan_\pi\AAA(\und{c}) = \bigotimes_{i=1}^n \hoRan_\pi\AAA(c_i)$, such that
the compatibility conditions \eqref{eqn:alg1}, \eqref{eqn:alg2} and \eqref{eqn:alg3} are satisfied.
Because the chain complex $\O_{\ovr{\CC}}\big(\substack{t\\ \und{c}}\big)$ is a free $k$-module
concentrated in degree $0$ with trivial differential (cf.\ \eqref{eqn:AQFToperadoperations}), 
we may equivalently specify $\Ch(k)$-morphisms
\begin{flalign}\label{eqn:alphasigmaf}
\alpha_{[\sigma,\und{f}]}^{} \,:\, \E_\infty(n) \otimes  \hoRan_\pi\AAA(\und{c})  \longrightarrow  \hoRan_\pi\AAA(t) \quad,
\end{flalign}
for all generators $[\sigma,\und{f}]\in \O_{\ovr{\CC}}\big(\substack{t\\ \und{c}}\big)$.
\sk

Using that the homotopy right Kan extension is computed by an end 
\eqref{eqn:hoRanend}, we will define \eqref{eqn:alphasigmaf} component-wise,
for all $y\in F(t)$ in the groupoid associated to the target color $t\in\CC$. To simplify our notation,
we abbreviate in what follows the normalized chain complexes 
in \eqref{eqn:hoRanend} by
\begin{flalign}\label{eqn:Nastshort}
N_\ast(x) \,:=\, N_\ast \big(B\big(F(c) \downarrow x\big)\big)\in\Ch(k)\quad,
\end{flalign}
for all  $x\in F(c)$. We will define the $y\in F(t)$ 
component of \eqref{eqn:alphasigmaf} by a $\Ch(k)$-morphism
\begin{flalign}\label{eqn:alphasigmafy}
\alpha_{[\sigma,\und{f}],y}^{} \,:\, \E_\infty(n)\otimes \bigotimes_{i=1}^n \big[N_\ast( f_i^\ast y),\AAA(c_i,f_i^\ast y)\big]
\longrightarrow \big[N_\ast(y),\AAA(t,y)\big]\quad,
\end{flalign}
where in the source we projected down from the ends to the $f_i^\ast y\in F(c_i)$ components
corresponding to the given $y\in F(t)$ and the given family of $\CC$-morphisms
$\und{f} = (f_1 :c_1\to t ,\dots , f_n : c_n\to t)$ that can be extracted from the generator $[\sigma,\und{f}]\in \O_{\ovr{\CC}}\big(\substack{t\\ \und{c}}\big)$ 
without ambiguity. Using further the 
adjunction $(-)\otimes V \dashv [V,-]$, for any $V\in\Ch(k)$, we may equivalently
consider the adjunct of the morphism \eqref{eqn:alphasigmafy}. 
Using also the symmetric braiding on $\Ch(k)$ to rearrange the tensor factors, 
we observe that defining \eqref{eqn:alphasigmafy} 
is equivalent to defining a $\Ch(k)$-morphism
\begin{flalign}
\widetilde{\alpha}_{[\sigma,\und{f}],y}^{} \,:\, 
\bigotimes_{i=1}^n \big[N_\ast(f_i^\ast y),\AAA(c_i,f_i^\ast y)\big]\otimes N_\ast(y)\otimes \E_\infty(n)
\longrightarrow \AAA(t,y)\quad.
\end{flalign}
Using that the normalized chain complex on a simplicial set carries a canonical {\em co}action of the 
$\E_\infty$-operad (cf.\ \cite[Theorem 2.1.1]{BergerFresse}), we define the
latter morphism by the composition of the following sequence of $\Ch(k)$-morphisms
\begin{flalign}\label{eqn:alphasigmafyexplicit}
\xymatrix{
\bigotimes\limits_{i=1}^n \big[N_\ast(f_i^\ast y),\AAA(c_i,f_i^\ast y)\big]\otimes N_\ast(y)\otimes \E_\infty(n)
\ar[d]^-{~\id\otimes \text{$\E_\infty$-coaction}}\\
\bigotimes\limits_{i=1}^n \big[N_\ast( f_i^\ast y),\AAA(c_i,f_i^\ast y)\big]\otimes N_\ast(y)^{\otimes n}
\ar[d]^-{~\id \otimes \Motimes_{i=1}^n N_\ast(f_i^\ast)} \\
\bigotimes\limits_{i=1}^n \big[N_\ast( f_i^\ast y),\AAA(c_i,f_i^\ast y)\big]\otimes 
\bigotimes\limits_{i=1}^n  N_\ast(f_i^\ast y)
\ar[d]^-{~\text{permute}}\\
\bigotimes\limits_{i=1}^n \Big(\big[N_\ast( f_i^\ast y),\AAA(c_i,f_i^\ast y)\big] \otimes N_\ast(f_i^\ast y) \Big)
\ar[d]^-{~\Motimes_{i=1}^n \ev}\\
\bigotimes\limits_{i=1}^n  \AAA(c_i,f_i^\ast y)
\ar[d]^-{~\text{$\O_{\ovr{\CC}_F}$-action of $[\sigma,(\und{f},\id_{\und{f}^\ast y})]$} }\\
\AAA(t,y)
}
\end{flalign}
Let us explain these steps in more detail: Step 1 uses the $\E_\infty$-coaction on normalized chains from
\cite[Theorem 2.1.1]{BergerFresse}. Recalling our notation in \eqref{eqn:Nastshort}, step 2 
is induced by the pullbacks $f_i^\ast : F(t)\downarrow y \to F(c_i)\downarrow f_i^\ast y $ of over categories
along the given $\CC$-morphisms $f_i : c_i\to t$. Step 3 is just a permutation of tensor factors
via the symmetric braiding. In step 4 we used the standard evaluation morphism $\ev : [V,W]\otimes V\to W$
for internal hom objects, which is the adjunct under $(-)\otimes V \dashv [V,-]$ of the identity $\id : [V,W]\to [V,W]$.
Finally, in step 5 we constructed the canonical lift
$[\sigma,(\und{f},\id_{\und{f}^\ast}y)] \in \O_{\ovr{\CC}_F}\big(\substack{ (t,y)\\ (\und{c},\und{f}^\ast y) }\big)$
of our given element $[\sigma,\und{f}] \in \O_{\ovr{\CC}}\big(\substack{t \\ \und{c}}\big)$ and used
the fact that $\AAA\in\QFT(\ovr{\CC}_F)\cong \Alg(\O_{\ovr{\CC}_F})$ is by hypothesis an $\O_{\ovr{\CC}_F}$-algebra.
\sk

Taking the adjunct of \eqref{eqn:alphasigmafyexplicit} defines \eqref{eqn:alphasigmafy}.
By a direct calculation one confirms that these components are compatible with the ends,
i.e.\ they define the morphism in \eqref{eqn:alphasigmaf} and consequently also 
\eqref{eqn:alphacomplete}. Confirming that the resulting morphisms
satisfy the compatibility conditions \eqref{eqn:alg1}, \eqref{eqn:alg2} and \eqref{eqn:alg3} 
is a rather lengthy but straightforward computation.
\end{proof}

\begin{rem}\label{rem:groupoidcoho}
The homotopy quantum field theory $\hoRan_\pi\AAA \in \QFT_w(\ovr{\CC})$ 
from Theorem \ref{theo:hoRanQFT} admits an interpretation in terms
of fiber-wise normalized cochain algebras on the category fibered in groupoids $\pi :\CC_F\to\CC$
with coefficients in a strict quantum field theory $\AAA \in \QFT(\ovr{\CC}_F)$ on the total category.
(In other words, this is the fiber-wise groupoid cohomology of $\pi :\CC_F\to\CC$ 
with coefficients in $\AAA \in \QFT(\ovr{\CC}_F)$.)
This can be understood by recalling that the chain complex $\hoRan_\pi\AAA(c)\in \Ch(k)$
assigned to an object $c\in \CC$ is given by the end formula \eqref{eqn:hoRanend},
where $F(c) \cong \pi^{-1}(c)\in \Grpd$ is the corresponding groupoid fiber. (Compare
this to the normalized cochains on a simplicial set with coefficients in $k$ given by 
\eqref{eqn:cochainssSetformula}.)
Similarly to ordinary groupoid cohomology \cite{Crainic}, these chain complexes
may have an interesting homology, even for the case where the input quantum
field theory $\AAA \in \QFT(\ovr{\CC}_F)$ is concentrated in degree $0$ and has a trivial
differential. Physical examples of the latter scenario have been discussed in \cite{BeniniSchenkel}
and they include e.g.\ Dirac fields on the groupoid of {\em all} possible spin structures over a spacetime.
\end{rem}


\section*{Acknowledgments}
We thank the anonymous referee for useful comments that helped us to improve this manuscript.
We also would like to thank Adrien Brochier, Klaus Fredenhagen, Ya{\"e}l Fr{\'e}gier, Owen Gwilliam, Dmitri Pavlov,
Birgit Richter, Claudia Scheimbauer, Urs Schreiber and Christoph Schweigert
for useful discussions and comments related to this work.
The work of M.B.\ is supported by a research grant funded by 
the Deutsche Forschungsgemeinschaft (DFG, Germany). 
A.S.\ gratefully acknowledges the financial support of 
the Royal Society (UK) through a Royal Society University 
Research Fellowship, a Research Grant and an Enhancement Award. 
L.W.\ is supported by the RTG 1670 ``Mathematics inspired 
by String Theory and Quantum Field Theory''.

\appendix

\section{\label{app:extension}Derived local-to-global extension for toy-models}
The goal of this appendix is to provide concrete evidence
that {\em derived} local-to-global extension functors
as in Examples \ref{ex:derivedFredenhagen} and \ref{ex:derivedBRSTBV}
are crucial for gauge theory.
In order to simplify our computations, we shall focus as in Remark \ref{rem:hAQFTlimitingcases}
and Section \ref{subsec:stacks} on very simple toy-models given by $\E_\infty$-algebra
valued functors. Because our concrete examples studied below are topological field
theories, in the sense that they are insensitive to the spacetime metric, we shall 
consider instead of the spacetime category $\Loc$ the category $\Man$ of oriented $m$-dimensional
manifolds of finite type with orientation preserving open embeddings as morphisms.
Hence, the relevant category of quantum field theories for this appendix
is given by $\QFT_{w}(\ovr{\Man}^{\max}) \cong \Alg(\E_\infty)^\Man$, where we
endowed $\Man$ with the maximal orthogonality relation from Remark \ref{rem:hAQFTlimitingcases}.
Instead of $\ovr{\Locc}$, we consider the full orthogonal subcategory
$\ovr{\Disc}^{\max}\subseteq \ovr{\Man}^{\max}$  of all manifolds diffeomorphic to $\bbR^m$.
The corresponding orthogonal embedding functor $j : \ovr{\Disc}^{\max} \to \ovr{\Man}^{\max}$
then defines a Quillen adjunction
\begin{flalign}\label{eqn:j!appendix}
\xymatrix{
j_! \,:\, \QFT_w(\ovr{\Disc}^{\max}) ~\ar@<0.5ex>[r]&\ar@<0.5ex>[l]  ~\QFT_w(\ovr{\Man}^{\max}) \,:\, j^\ast
}\quad,
\end{flalign}
whose left adjoint is an extension functor to all manifolds 
for theories $\BBB \in \QFT_w(\ovr{\Disc}^{\max})\cong\Alg(\E_\infty)^\Disc$ 
that are only defined on disks.
\sk

The main technical tool that we use to compute the derived local-to-global extension
of our examples below is the simplicial set tensoring for $\E_\infty$-algebras
\cite{fressecotripel, gtz}. This is a functor
$\otimes :\sSet \times \Alg(\E_\infty)\to \Alg(\E_\infty)$ that assigns to
a simplicial set $K\in \sSet$ and an $\E_\infty$-algebra $A\in \Alg(\E_\infty)$
a new $\E_\infty$-algebra $K\otimes A \in \Alg(\E_\infty)$.
The corresponding derived tensoring is given by 
\begin{flalign}
K  \stackrel{\bbL}{\otimes}  A \,:=\,  K\otimes Q(A)\quad,
\end{flalign}
where $Q(A)$ is a cofibrant replacement of $A$ in the model category 
$\Alg(\E_\infty)$. The latter is also referred to as the derived higher 
Hochschild chains on $K$ with coefficients in $A$, cf.\ \cite{gtz}. 
\sk

The following technical theorem will be the key ingredient for
the concrete computations of derived local-to-global extensions in the examples below.
\begin{theo}\label{theo:hoLan}
Let $\AAA\in \QFT_w(\ovr{\Disc}^{\max})$ be naturally weakly equivalent to a constant 
functor $\Disc\to\Alg(\E_\infty)$ whose value we denote by $A\in \Alg(\E_\infty)$.
Then the derived local-to-global extension 
$\bbL j_!(\AAA) \in \QFT_w(\ovr{\Man}^{\max})$
may be computed object-wise for $M\in\Man$ by
\begin{flalign}\label{eqn:exttensor}
\big(\mathbb{L}j_! (\AAA)\big)(M) \,=\, \Sing(M) \stackrel{\mathbb{L}}{\otimes} A \quad,
\end{flalign}
where $\Sing(M)\in\sSet$ denotes the simplicial set of singular simplices in $M$. 
\end{theo}
\begin{proof}
Because of $\QFT_w(\ovr{\Man}^{\max}) \cong \Alg(\E_{\infty})^\Man$ 
and $\QFT_w(\ovr{\Disc}^{\max}) \cong \Alg(\E_{\infty})^\Disc$,
the left adjoint $j_!(\AAA)$ in \eqref{eqn:j!appendix} is given
by left Kan extension of $\E_\infty$-algebra valued functors along
the embedding functor $j: \Disc\to \Man$. The left derived functor $\bbL j_!(\AAA)$
is therefore a homotopy left Kan extension, which by \cite[Proposition~1.14]{hoKan2}
can be computed object-wise for $M\in\Man$ as a homotopy colimit of the diagram
$\Disc \downarrow M \to \Disc \xrightarrow{\AAA} \Alg(\E_\infty)$. 
Using that by hypothesis $\AAA$ is naturally weakly equivalent to a constant functor
with value $A\in\Alg(\E_\infty)$, it follows from \cite[Chapter 19]{Hirschhorn}
that this homotopy colimit is given by the derived tensoring
\begin{flalign}
\big(\bbL j_!(\AAA)\big)(M) \,\simeq\, B(\Disc \downarrow M) \stackrel{\mathbb{L}}{\otimes} A\quad,
\end{flalign}
where $B:\Cat\to \sSet$ is the nerve functor. 
\sk

To conclude the proof, it is sufficient to show that $B(\Disc \downarrow M) \simeq \Sing(M)$
are weakly equivalent in $\sSet$ because the derived tensoring preserves (by construction) weak equivalences.
Let us denote by $\Open_M$ the category of open subsets of $M$
and introduce the functor 
\begin{flalign}
\chi\,:\, \Disc\downarrow M \, \longrightarrow \, \Open_M
\end{flalign}
that sends an open embedding $(f : U \to M) \in \Disc \downarrow M$ 
of a disk $U$ into $M$ to its image $f(U)\in \Open_M$. For each point $x\in M$, 
consider the full subcategory $(\Disc \downarrow M)_x \subseteq \Disc \downarrow M$ 
of objects whose images under $\chi$ contain the point $x$. 
One easily observes that $(\Disc \downarrow M)_x$ is cofiltered, i.e.\ for 
any two open subsets $U_1, U_2 \subseteq M$ containing $x$, there exists 
a smaller open subset $U \subseteq U_1\cap U_2\subseteq M$ that is diffeomorphic to $\bbR^m$ and 
contains $x$. By \cite[Section 1]{Quillen}, this implies that the nerve of $(\Disc \downarrow M)_x$
is a contractible simplicial set. With this observation we can apply Lurie's 
Seifert-van~Kampen~theorem \cite[A.3.1]{HA} to the functor $\chi$ 
and conclude that there is a weak equivalence (in $\sSet$)
\begin{flalign}
\hocolim\big(\Sing \circ \chi \big) \,\simeq\, \Sing(M)\quad,
\end{flalign}
where $\chi$ is regarded here as a functor with values in topological spaces 
by equipping each open subset of $M$ with the induced topology. 
Since $\chi(f: U \to M)$ is a contractible topological space, 
for each $(f: U \to M) \in \Disc \downarrow M$, it follows that 
$\Sing \circ \chi: \Disc \downarrow M \to \sSet$ is naturally weakly 
equivalent to the constant functor $\ast: \Disc \downarrow M \to \sSet$ 
whose value is the terminal object $\ast \in \sSet$.
Summing up, we obtained the following chain of weak equivalences (in $\sSet$)
\begin{flalign}
\Sing(M) \,\simeq\, \hocolim\big(\Sing \circ \chi\big) 
\, \simeq \, \hocolim\big(\ast: \Disc \downarrow M \to \sSet\big)
\,\simeq\, B(\Disc \downarrow M) \quad.
\end{flalign}
This concludes the proof.
\end{proof}

We shall now discuss two explicit toy-models that are inspired by
topological field theory, namely simplified versions of Dijkgraaf-Witten theory 
and Chern-Simons theory. The aim of our first example is
to show that the derived local-to-global extension $\bbL j_!$ captures the
desired global gauge theory observables on topologically non-trivial manifolds,
while the underived extension functor fails to do so. Hence, the following example provides a
concrete argument why the derived functor is preferable also from a physical perspective.
\begin{ex}
Inspired by Dijkgraaf-Witten theory, we consider a field theory
whose fields on a manifold $M$ are described by the groupoid $\PBun_G(M)$
of principal $G$-bundles on $M$, where $G$ is a finite nilpotent group (e.g.\ Abelian). 
Using \cite[Lemma~2.8]{SW} (and its proof), we can describe (up to weak equivalence in $\sSet$)
the nerve $B\PBun_G(M)$ of this groupoid by the simplicial mapping space 
$BG^{\Sing(M)}$, where $BG$ is the nerve of the groupoid associated to the group $G$. 
In complete analogy to Corollary \ref{cor:stackQFTs}, this defines
a theory $\AAA \in \QFT_w(\ovr{\Man}^{\max})\cong \Alg(\E_\infty)^{\Man}$
that assigns to each $M\in\Man$ the normalized cochain algebra
\begin{flalign}\label{eqn:DW}
\AAA(M)\,:=\, N^\ast\big(BG^{\Sing(M)}\big) \in\Alg(\E_\infty)
\end{flalign}
on the groupoid of fields.
\sk

Our goal is to show that the {\em derived} local-to-global
extension $\bbL j_!\, j^\ast(\AAA)$ of the restriction
$j^\ast(\AAA)\in \QFT_w(\ovr{\Disc}^{\max})\cong \Alg(\E_\infty)^{\Disc}$
to disks is weakly equivalent to the original theory $\AAA$.
For this we observe that $j^\ast(\AAA)$ is naturally weakly equivalent 
to the constant functor $N^\ast(BG) : \Disc \to \Alg(\E_\infty)$
that assigns the normalized cochain algebra on $BG$. 
(This is because $\Sing(U)\to \ast$ is a weak equivalence in $\sSet$ 
as the underlying topological space of every $U\in\Disc$ is contractible.)
Hence, the hypotheses of Theorem \ref{theo:hoLan} are fulfilled
and we obtain a chain of weak equivalences (in $\Alg(\E_\infty)$)
\begin{flalign}
\big(\bbL j_! \,j^\ast(\AAA)\big)(M) \,\simeq\, \Sing(M) \stackrel{\bbL}{\otimes} N^\ast(BG) \,\simeq\, N^\ast(BG^{\Sing(M)})\quad,
\end{flalign}
where the last step was proven in \cite[Proposition~5.3]{AyalaFrancis}
and it uses the assumption that $G$ is nilpotent. We conclude that $\mathbb{L}j_!\, j^\ast(\AAA) \simeq \AAA$. 
\sk

Let us discuss these results from a more physical perspective.
For a manifold $M\in\Man$, the $\E_\infty$-algebra $\AAA(M)$ 
is a ``higher function algebra'' on the groupoid $\PBun_G(M)$ of principal
$G$-bundles on $M$. Its 0th homology
$H_0(\AAA(M)) = \mathrm{Map}\big(\pi_0\PBun_G(M) , k\big)$ 
is the ordinary algebra of functions on the set $\pi_0\PBun_G(M)$ of  isomorphism
classes of such bundles. This means that the $\E_\infty$-algebra $\AAA(M)$ 
describes observables that in particular can be used to distinguish 
different principal $G$-bundles on $M$, which is of course what 
Dijkgraaf-Witten theory is about. At first sight, it seems that the
restriction $j^\ast(\AAA)$ to disks discards such bundle observables
because $\pi_0\PBun_G(U) \cong\ast$ is the point and hence the homology
$H_0(\AAA(U))\cong k$ is the trivial algebra, for all $U\in\Disc$.
This is however not true: as we have proven above, the {\em derived}
local-to-global extension $\bbL j_!\, j^\ast(\AAA)$ recovers the bundle
observables from the higher algebraic structure of $j^\ast(\AAA)$,
or equivalently of the constant functor $N^\ast(BG) : \Disc \to \Alg(\E_\infty)$.
In stark contrast to this, the {\em underived} local-to-global extension 
is {\em not} capable to recover such bundle observables because it 
extends our constant functor to the functor $j_!\, j^\ast(\AAA) : 
\Man\to  \Alg(\E_\infty)$ that assigns to each $M \in \Man$ 
the coproduct $\bigsqcup_{\pi_0(M)} N^\ast(BG)$ of $\E_{\infty}$-algebras, 
where $\pi_0(M)$ denotes the set of connected components of $M$. 
Therefore, for $M$ connected, the 0th homology of 
$j_!\, j^\ast(\AAA)(M)$ is the trivial algebra $k$.
\end{ex}

In our second example we prove that the {\em underived}  
local-to-global extension $j_!$ is not compatible with weak equivalences.
More precisely, we show that, even when it happens that $j_!$
produces the desired result for a certain choice of input, changing
to a weakly equivalent input results in a non-equivalent outcome.
This issue is solved (by construction) by the {\em derived}
extension $\bbL j_!$.
\begin{ex}
Inspired by Chern-Simons theory with structure group $\bbR$,
we consider a field theory whose fields on a $2$-dimensional 
manifold $M$ are described by the chain complex of principal $\bbR$-bundles
with flat connections on $M$. (Throughout the whole 
example $\Man$ will be the category of $2$-dimensional
oriented manifolds.) The chain complex of linear
observables on $M\in \Man$ for this theory is concretely defined 
by the $(-1)$-shifted compactly-supported de Rham complex, i.e.\
\begin{flalign}
\Omega_{c}^\ast(M)[-1]\, :=\, \bigg(   
\xymatrix@C=1.5em{
\dots & \ar[l] 0  & \ar[l] \stackrel{(-1)}{\Omega_c^2(M) } & \ar[l]\stackrel{(0)}{\Omega_c^1(M)}
& \ar[l] \stackrel{(1)}{\Omega_c^0 (M)} & \ar[l] 0 & \ar[l] \cdots
}\bigg)\quad,
\end{flalign}
where we indicated in round brackets our homological degree conventions.
We define a theory $\AAA \in \QFT_w({\ovr{\Man}}^{\max}) \cong \Alg(\E_{\infty})^\Man$
by assigning to each $M\in\Man$ the free $\E_{\infty}$-algebra
\begin{flalign}
\AAA(M) \,:= \, \bbE_{\infty}\big(\Omega_{c}^\ast(M)[-1]\big) \in\Alg(\E_\infty)
\end{flalign}
over $\Omega_{c}^\ast(M)[-1]$. (This is a homotopy-coherent 
analogue of the strictly commutative symmetric algebra over a chain complex.)
\sk

As in the previous example, our first goal is to show that the {\em derived} local-to-global
extension $\bbL j_!\, j^\ast(\AAA)$ of the restriction
$j^\ast(\AAA)\in \QFT_w(\ovr{\Disc}^{\max})\cong \Alg(\E_\infty)^{\Disc}$
to disks is weakly equivalent to the original theory $\AAA$.
For this we observe that $j^\ast(\AAA)$ is naturally weakly equivalent
to the constant functor $\bbE_{\infty}(\mathbb{R}[1]) : \Disc \to \Alg(\E_\infty)$ 
that assigns the free $\E_\infty$-algebra over the $1$-shifted monoidal unit $\bbR[1]$. 
(This is because the integration map $\int_U : \Omega_{c}^\ast(M)[-1]\to \bbR[1]$ is a 
weak equivalence between cofibrant chain complexes for every $U\in\Disc$.)
Theorem \ref{theo:hoLan} then implies
\begin{flalign}
\big(\bbL j_!\,j^\ast(\AAA)\big)(M) \,\simeq\,  
\Sing(M) \stackrel{\mathbb{L}}{\otimes} \bbE_{\infty}(\mathbb{R}[1])\quad.
\end{flalign}
Because $\bbE_{\infty}(\bbR[1]) \in \Alg(\E_{\infty})$ is a free $\E_{\infty}$-algebra, 
the derived tensoring is weakly equivalent to the underived one. Using the explicit
formula of \cite{fressecotripel}, a direct computation then shows that
\begin{flalign}
\big(\bbL j_!\,j^\ast(\AAA)\big)(M) \,\simeq \,\Sing(M) \otimes \bbE_\infty(\mathbb{R}[1]) 
\,=\, \bbE_{\infty}\big(N_\ast(\Sing(M)) \otimes \mathbb{R}[1]\big) \quad .
\end{flalign} 
One concludes that $\mathbb{L}j_!\, j^\ast(\AAA) \simeq \AAA$
because the $1$-shifted $\bbR$-valued normalized chain complex
$N_\ast(\Sing(M)) \otimes \mathbb{R}[1]$ is naturally weakly equivalent to 
$\Omega_{c}^\ast(M)[-1]$, which is a consequence of de Rham's theorem.
\sk

Let us now compute the {\em underived}
local-to-global extension $j_!\, j^\ast (\AAA)$ of the restriction $j^\ast(\AAA)$,
which is given by left Kan extension along $j : \Disc\to\Man$.
Because $j_!$ can be computed object-wise as a colimit
and the left adjoint functor $\bbE_\infty$ preserves colimits,
one finds that
\begin{flalign}\label{eqn:tmpunderived1}
\big(j_!\,j^\ast(\AAA)\big)(M) \,\cong\, 
\bbE_\infty \Big( \colim \Big( \Disc \downarrow M 
\to \Disc \xrightarrow{\Omega^\ast_c[-1]} \Ch(\bbR) \Big) \Big) 
\,\cong\, \bbE_\infty\big(\Omega^\ast_c[-1](M)\big)\quad,
\end{flalign}
where in the last step one exploits a partition of unity
argument on $M$ to explicitly compute the colimit, cf.\ \cite[Lemma 4.4.1]{Lang}. 
Note that in this specific situation the underived extension $j_!\,j^\ast(\AAA)$
is naturally isomorphic to the original theory $\AAA$, which is however accidental as this result 
is not stable under weak equivalence. In particular, if we replace the restriction 
$j^\ast (\AAA)$ by the naturally weakly equivalent constant functor 
$\bbE_{\infty}(\mathbb{R}[1]): \Disc\to \Alg(\E_\infty)$, then another
direct calculation shows that
\begin{flalign}\label{eqn:tmpunderived2}
\big(j_!\big(\bbE_{\infty}(\mathbb{R}[1])\big)\big)(M)\,\cong\,\bbE_\infty \Big( \colim \Big( \Disc \downarrow M 
\to \Disc \xrightarrow{\bbE_\infty(\bbR[1])} \Ch(\bbR) \Big) \Big) 
\,\cong\, \bigsqcup_{\pi_0(M)} \bbE_\infty ( \bbR[1]) \quad,
\end{flalign}
where $\pi_0(M)$ is the set of connected components of $M$.
One easily sees that \eqref{eqn:tmpunderived1} and \eqref{eqn:tmpunderived2} 
are not weakly equivalent by considering for example the $2$-sphere $M= \bbS^2$.
The conclusion is that the {\em underived} local-to-global extension $j_!$
is inconsistent in our homotopical framework because 
it assigns to weakly equivalent theories $j^\ast(\AAA) \simeq \bbE_\infty(\bbR[1])$ 
on disks genuinely different theories on all manifolds. This lack of stability under weak equivalences
is of course solved (by construction) by working with the {\em derived} extension functor $\bbL j_!$.
\end{ex}



\begin{thebibliography}{10}


\bibitem[AF15]{AyalaFrancis} 
D.~Ayala and J.~Francis,
``Factorization homology of topological manifolds,''
J.\ Topol.\ {\bf 8}, no.\ 4, 1045--1084 (2015)
[arXiv:1206.5522 [math.AT]].



\bibitem[BS17]{BeniniSchenkel} 
M.~Benini and A.~Schenkel,
``Quantum field theories on categories fibered in groupoids,''
Commun.\ Math.\ Phys.\  {\bf 356}, no.\ 1, 19 (2017)
[arXiv:1610.06071 [math-ph]].


\bibitem[BSS18]{BSSStack} 
M.~Benini, A.~Schenkel and U.~Schreiber,
``The stack of Yang-Mills fields on Lorentzian manifolds,''
Commun.\ Math.\ Phys.\  {\bf 359}, no.\ 2, 765 (2018)
[arXiv:1704.01378 [math-ph]].
  

\bibitem[BSS15]{BeniniSchenkelSzabo} 
M.~Benini, A.~Schenkel and R.~J.~Szabo,
``Homotopy colimits and global observables in Abelian gauge theory,''
Lett.\ Math.\ Phys.\  {\bf 105}, no.\ 9, 1193 (2015)
[arXiv:1503.08839 [math-ph]].
  

\bibitem[BSW17]{BeniniSchenkelWoike} 
M.~Benini, A.~Schenkel and L.~Woike,
``Operads for algebraic quantum field theory,''
arXiv:1709.08657 [math-ph].


\bibitem[BSW18]{BSWinvolutions}
M.~Benini, A.~Schenkel and L.~Woike,
``Involutive categories, colored $\ast$-operads and quantum field theory,''
arXiv:1802.09555 [math.CT].


\bibitem[BF04]{BergerFresse}
C.~Berger and B.~Fresse,
``Combinatorial operad actions on cochains,''
Math.\ Proc.\ Cambridge Philos.\ Soc.\ {\bf 137}, no.\ 1, 135--174 (2004)
[arXiv:math/0109158 [math.AT]]. 


\bibitem[BM03]{BergerMoerdijkOriginal}
C.~Berger and I.~Moerdijk,
``Axiomatic homotopy theory for operads,'' 
Comment.\ Math.\ Helv.\ {\bf 78}, 805--831 (2003)
[arXiv:math/0206094 [math.AT]].


\bibitem[BM07]{BergerMoerdijk}
C.~Berger and I.~Moerdijk,
``Resolution of coloured operads and rectification of homotopy algebras,''
in: A.~Davydov, M.~Batanin, M.~Johnson, S.~Lack and A.~Neeman (eds.),
{\it Categories in algebra, geometry and mathematical physics},
Contemp.\ Math.\ {\bf 431}, 31--58,  
American Mathematical Society, Providence, RI (2007). 


\bibitem[BDFY15]{AQFTbook}
R.~Brunetti, C.~Dappiaggi, K.~Fredenhagen and J.~Yngvason,
{\it Advances in algebraic quantum field theory},
Springer Verlag, Heidelberg (2015).


\bibitem[BFV03]{Brunetti} 
R.~Brunetti, K.~Fredenhagen and R.~Verch,
``The generally covariant locality principle: A new paradigm for local quantum field theory,''
Commun.\ Math.\ Phys.\  {\bf 237}, 31 (2003)
[math-ph/0112041].


\bibitem[CPTVV17]{DAG2}
D.~Calaque, T.~Pantev, B.~To{\"e}n, M.~Vaqui{\'e} and G.~Vezzosi,
``Shifted Poisson structures and deformation quantization,''
J.\ Topol.\ {\bf 10}, no.\ 2, 483--584  (2017)
[arXiv:1506.03699 [math.AG]].

  
\bibitem[Cis09]{hoKan2}  
D.-C.~Cisinski, 
``Locally constant functors,''
Math.\ Proc.\ Cambridge Philos.\ Soc.\ {\bf 147}, no.\ 3, 593--614  (2009)
[arXiv:0803.4342 [math.AT]].
 

\bibitem[CM13]{CisinskiMoerdijk}
D.-C.~Cisinski and I.~Moerdijk, 
``Dendroidal sets and simplicial operads,'' 
J.\ Topol.\ {\bf 6}, no.\ 3, 705--756 (2013) 
[arXiv:1109.1004 [math.AT]].


\bibitem[CG17]{CostelloGwilliam}
K.~Costello and O.~Gwilliam,
{\it Factorization algebras in quantum field theory. Vol.\ 1},
New Mathematical Monographs {\bf 31}, 
Cambridge University Press, Cambridge (2017).


\bibitem[Cra03]{Crainic}
M.~Crainic, 
``Differentiable and algebroid cohomology, van Est isomorphisms, and characteristic classes,''
Comment.\ Math.\ Helv.\ {\bf 78}, no.\ 4, 681--721 (2003).


\bibitem[Dou17]{Dougherty}
J.~Dougherty, 
``Sameness and Separability in Gauge Theories,''
Philosophy of Science {\bf 84}, no.\ 5, 1189--1201 (2017).


\bibitem[Dug01]{DuggerUniversal}
D.~Dugger,
``Universal homotopy theories,''
Adv.\ Math.\ {\bf 164}, no.\ 1, 144--176  (2001)
[arXiv:math/0007070 [math.AT]]. 
  

\bibitem[DHI04]{Dugger}
D.~Dugger, S.~Hollander and D.~C.~Isaksen,
``Hypercovers and simplicial presheaves,''
Math.\ Proc.\ Cambridge Philos.\ Soc.\ {\bf 136}, no.\ 1, 9--51 (2004)
[arXiv:math/0205027 [math.AT]]. 

  
\bibitem[DS95]{Dwyer}
W.~G.~Dwyer and J.~Spalinski,
``Homotopy theories and model categories,''
in: I.~M.~James (ed.), {\it Handbook of Algebraic Topology}, 
73, North-Holland, Amsterdam (1995). 
   

\bibitem[FV15]{FewsterVerch}
C.~J.~Fewster and R.~Verch,
``Algebraic quantum field theory in curved spacetimes,''
in: R.~Brunetti, C.~Dappiaggi, K.~Fredenhagen and J.~Yngvason (eds.),
{\it Advances in algebraic quantum field theory}, 125--189,
Springer Verlag, Heidelberg (2015)
[arXiv:1504.00586 [math-ph]].


\bibitem[Fre90]{Fre1}
K.~Fredenhagen, 
``Generalizations of the theory of superselection sectors,''
in: D.~Kastler (ed.),
{\it The algebraic theory of superselection sectors: Introduction and recent results}, 
World Scientific Publishing, 379 (1990).


\bibitem[Fre93]{Fre2}
K.~Fredenhagen, 
``Global observables in local quantum physics,''
in: H.~Araki, K.~R.~Ito, A.~Kishimoto and I.~Ojima (eds.),
{\it Quantum and non-commutative analysis: Past, present and future perspectives}, 41--51,
Kluwer Academic Publishers (1993).
  

\bibitem[FRS92]{Fre3}
K.~Fredenhagen, K.-H.~Rehren and B.~Schroer, 
``Superselection sectors with braid group statistics and exchange algebras II: 
Geometric aspects and conformal covariance,'' 
Rev.\ Math.\ Phys.\ {\bf 4}, 113 (1992).


\bibitem[FR12]{FredenhagenRejzner} 
K.~Fredenhagen and K.~Rejzner,
``Batalin-Vilkovisky formalism in the functional approach to classical field theory,''
Commun.\ Math.\ Phys.\  {\bf 314}, 93 (2012)
[arXiv:1101.5112 [math-ph]].
  
  
\bibitem[FR13]{FredenhagenRejzner2} 
K.~Fredenhagen and K.~Rejzner,
``Batalin-Vilkovisky formalism in perturbative algebraic quantum field theory,''
Commun.\ Math.\ Phys.\  {\bf 317}, 697 (2013)
[arXiv:1110.5232 [math-ph]].


\bibitem[Fre16]{fressecotripel}
B.~Fresse,
``The cotriple resolution of differential graded algebras,''
Proc.\ Amer.\ Math.\ Soc.\ {\bf 144}, 4693--4707  (2016)
[arXiv:1503.08489 [math.AT]]


\bibitem[GTZ14]{gtz}
G.~Ginot, T.~Tradler and M.~Zeinalian,
``Derived Higher Hochschild Homology, Topological Chiral Homology and Factorization Algebras,''
Commun.\ Math.\ Phys.\ {\bf 326}, 635--686 (2014).
[arXiv:1011.6483 [math.QA]]


\bibitem[GJ99]{GoerssJardine}
P.~G.~Goerss and J.~F.~Jardine,
{\it Simplicial homotopy theory},
Progress in Mathematics {\bf 174},
Birkh\"auser Verlag, Basel (1999).


\bibitem[GH18]{BVquantization}
O.~Gwilliam and R.~Haugseng,
``Linear Batalin-Vilkovisky quantization as a functor of $\infty$-categories,''
Selecta Math.\ {\bf 24}, no.\ 2, 1247--1313  (2018)
[arXiv:1608.01290 [math.AT]].


\bibitem[HK64]{HaagKastler}
R.~Haag and D.~Kastler,
``An algebraic approach to quantum field theory,''
J.\ Math.\ Phys.\  {\bf 5}, 848 (1964).


\bibitem[Hin97]{HinichOriginal}
V.~Hinich, 
``Homological algebra of homotopy algebras,'' 
Comm.\ Algebra {\bf 25}, no.\ 10, 3291--3323  (1997)
[arXiv:q-alg/9702015].
Erratum: arXiv:math/0309453 [math.QA].


\bibitem[Hin15]{Hinich}
V.~Hinich, 
``Rectification of algebras and modules,'' 
Doc.\ Math.\ {\bf 20}, 879--926 (2015)
[arXiv:1311.4130 [math.QA]].


\bibitem[Hir03]{Hirschhorn}
P.~S.~Hirschhorn,
{\it Model categories and their localizations},
Math.\ Surveys Monogr.\ {\bf 99}, 
Amer.\ Math.\ Soc., Providence, RI (2003).


\bibitem[Hol08a]{Hollander}
S.~Hollander,
``A homotopy theory for stacks,''
Israel J.\ Math.\ {\bf 163}, 93--124 (2008)
[arXiv:math.AT/0110247].


\bibitem[Hol08b]{Hollands} 
S.~Hollands,
``Renormalized Quantum Yang-Mills Fields in Curved Spacetime,''
Rev.\ Math.\ Phys.\  {\bf 20}, 1033 (2008)
[arXiv:0705.3340 [gr-qc]].
  
  
\bibitem[Hov99]{Hovey}
M.~Hovey,
{\it Model categories}, 
Math.\ Surveys Monogr.\ {\bf 63}, 
Amer.\ Math.\ Soc., Providence, RI (1999).


\bibitem[Kaw15]{Kawahigashi}
Y.~Kawahigashi,
``Conformal Field Theory, Tensor Categories and Operator Algebras,''
J.\ Phys.\ A {\bf 48}, no.\ 30, 303001 (2015) 
[arXiv:1503.05675 [math-ph]].


\bibitem[Lan14]{Lang}
B.~Lang, 
{\it Universal constructions in algebraic 
and locally covariant quantum field theory}, 
Ph.D.~thesis, U.~York (2014), 
available at \url{http://etheses.whiterose.ac.uk/8019/}.


\bibitem[LV12]{LodayVallette}
J.-L.~Loday and B.~Vallette,
{\it Algebraic operads},
Grundlehren der Mathematischen Wissenschaften {\bf 346},
Springer Verlag, Heidelberg (2012).


\bibitem[Lur09]{TQFT}
J.~Lurie, 
``On the Classification of Topological Field Theories,'' 
Current Developments in Mathematics, Volume 2008, 129--280 (2009)
[arXiv:0905.0465 [math.CT]].


\bibitem[LurHTT]{Lurie}
J.~Lurie,
{\it Higher topos theory}, 
Annals of Mathematics Studies {\bf 170}, 
Princeton University Press (2009).


\bibitem[LurHA]{HA}
J.~Lurie,
{\it Higher algebra},
book draft available at
\url{http://www.math.harvard.edu/~lurie/papers/HA.pdf}


\bibitem[NTW17]{Teh}
J.~Nguyen, N.~J.~Teh and L.~Wells,
``Why surplus structure is not superfluous,''
{\em to appear in The British Journal for the Philosophy of Science}
[arXiv:1712.01228 [physics.hist-ph]].


\bibitem[PTVV13]{DAG}
T.~Pantev, B.~To{\"e}n, M.~Vaqui{\'e} and G.~Vezzosi,
``Shifted symplectic structures,''
Publ.\ Math.\ Inst.\ Hautes {\'E}tudes Sci.\ {\bf 117}, 271--328 (2013)
[arXiv:1111.3209 [math.AG]]. 


\bibitem[Qui10]{Quillen}
D.~Quillen, 
``Higher algebraic $K$-theory: I,'' 
in: L.~Ji, K.~Liu and S.-T.~Yau (eds.), 
{\it Cohomology of groups and algebraic $K$-theory}, 
Adv.\ Lect.\ Math.\ {\bf 12}, 
Int.\ Press, Somerville, MA (2010). 


\bibitem[PS18]{PavlovScholbach}
D.~Pavlov and J.~Scholbach,
``Admissibility and rectification of colored symmetric operads,''
Journal of Topology {\bf 11}, 559--601 (2018)
[arXiv:1410.5675 [math.AT]].


\bibitem[Rie14]{Riehl}
E.~Riehl,
{\it Categorical homotopy theory},
New Mathematical Monographs {\bf 24},
Cambridge University Press, Cambridge (2014).


\bibitem[Rod14]{hoKan3}
B.~Rodr{\'i}guez-Gonz{\'a}lez,
``Realizable homotopy colimits,''
Theory and Applications of Categories {\bf 29}, no.\ 22, 609--634 (2014)
[arXiv:1104.0646 [math.AG]].


\bibitem[Sch13]{Schreiber}
U.~Schreiber,
``Differential cohomology in a cohesive infinity-topos,''
current version available at \url{https://ncatlab.org/schreiber/show/differential+cohomology+in+a+cohesive+topos}
[arXiv:1310.7930 [math-ph]].


\bibitem[SS03]{Schwede}
S.~Schwede and B.~Shipley,
``Equivalences of monoidal model categories,''
Algebr.\ Geom.\ Topol. {\bf 3}, 287--334 (2003)
[arXiv:math/0209342 [math.AT]]. 


\bibitem[SW18]{SW}
  C.~Schweigert and L.~Woike,
  ``Extended Homotopy Quantum Field Theories and their Orbifoldization,''
  arXiv:1802.08512 [math.QA].
  

\bibitem[Spi01]{Spitzweck}
M.~Spitzweck,
{\it Operads, algebras and modules in general model categories},
PhD Thesis, Bonn (2001).
Available at \url{http://hss.ulb.uni-bonn.de/2001/0241/0241.pdf}


\bibitem[TZ18]{TehraniZahn} 
M.~T.~Tehrani and J.~Zahn,
``Background independence in gauge theories,''
arXiv:1804.07640 [math-ph].


\bibitem[Yau16]{Yau}
D.~Yau,
{\it Colored operads},
Graduate Studies in Mathematics {\bf 170},
American Mathematical Society, Providence, RI (2016).


\bibitem[Yau18]{YauQFT}
D.~Yau,
``Homotopical Quantum Field Theory,''
arXiv:1802.08101 [math-ph].


\bibitem[Zah14]{Zahn}
J.~Zahn,
``The renormalized locally covariant Dirac field,''
Rev.\ Math.\ Phys.\  {\bf 26}, no.\ 1, 1330012 (2014)
[arXiv:1210.4031 [math-ph]].
  

\end{thebibliography}
\end{document}